\documentclass[journal]{IEEEtran}

\usepackage{cite}
\usepackage{graphicx}
\usepackage[tight,footnotesize]{subfigure}
\usepackage{array}
\usepackage{amsmath}
\usepackage{amssymb}
\usepackage{subeqnarray}
\usepackage{cases}
\usepackage{mdwtab}
\usepackage{algorithm}
\usepackage{algorithmic}
\usepackage{multirow}
\usepackage{xcolor}
\usepackage{epsf}
\usepackage{epstopdf}
\usepackage{epsfig}

\newtheorem{theorem}{Theorem}
\newtheorem{lemma}{Lemma}
\newtheorem{proposition}{Proposition}
\newtheorem{corollary}{Corollary}

\newtheorem{remark}{Remark}

\begin{document}
%
\title{Outage Minimization for a Fading Wireless Link with Energy Harvesting Transmitter and Receiver}


\author{Sheng~Zhou,~\IEEEmembership{Member,~IEEE}, Tingjun~Chen, Wei~Chen,~\IEEEmembership{Senior Member,~IEEE}, Zhisheng~Niu,~\IEEEmembership{Fellow,~IEEE}
\thanks{This work is sponsored in part by the National Science Foundation of China (NSFC) under grant No. 61201191, the National Basic Research Program of China (973 Program: No. 2012CB316001 and No. 2013CB336600), the NSFC under grant No. 61322111, No. 61321061, No. 61401250 and No. 61461136004, and Hitachi Ltd. Part of this work has been presented at the International Symposium on Modeling and Optimization in Mobile, Ad Hoc and Wireless Networks (WiOpt), May 2014 \cite{Chen14}.

Sheng Zhou, Wei Chen and Zhisheng Niu are with Tsinghua National Laboratory for Information Science and Technology, Dept. of
Electronic Engineering, Tsinghua University, Beijing 100084, China. Email:
\{sheng.zhou, wchen, niuzhs\}@tsinghua.edu.cn.

Tingjun Chen was previously with Dept. of Electronic Engineering, Tsinghua University, Beijing, China. He is now with the Department of Electrical
Engineering, Columbia University, New York, NY, USA. Email: tingjun@ee.columbia.edu.
}}

\maketitle

\begin{abstract}

    This paper studies online power control policies for outage minimization in a fading wireless link with energy harvesting transmitter and receiver. The outage occurs when either the transmitter or the receiver does not have enough energy, or the channel is in outage, where the transmitter only has the channel distribution information. Under infinite battery capacity and without retransmission, we prove that threshold-based power control policies are optimal. We thus propose disjoint/joint threshold-based policies with and without battery state sharing between the transmitter and receiver, respectively. We also analyze the impact of practical receiver detection and processing on the outage performance. When retransmission is considered, policy with linear power levels is adopted to adapt the power thresholds per retransmission. With finite battery capacity, a three dimensional finite state Markov chain is formulated to calculate the optimal parameters and corresponding performance of proposed policies. The energy arrival correlation between the transmitter and receiver is addressed for both finite and infinite battery cases. Numerical results show the impact of battery capacity, energy arrival correlation and detection cost on the outage performance of the proposed policies, as well as the tradeoff between the outage probability and the average transmission times.

\end{abstract}


\section{Introduction}
Wireless transceivers powered by renewable energy are becoming more and more appealing due to their ease of deployment and environment friendliness \cite{Zorzi14}, of which the feasibility is confirmed by measurements \cite{Zussman11}. When the transceivers in wireless links are powered by energy harvesting (EH), due to the randomness of energy arrivals, the transmission can fail when the energy stored in the battery is insufficient, and thus transmission reliability and throughput are degraded. This calls for new ways of power control policies adapting themselves to not only the channel fading, but also random energy arrivals at \emph{both} the transmitter and the receiver.

There have been some recent studies about EH-based wireless transmissions, specifically dealing with the multi-fold randomness from energy arrival, data arrival and channel fading, and their main focuses are on EH transmitters. With infinite battery capacity, optimal power control policies are studied in \cite{Sharma10} and \cite{Thompson12} to stabilize the system under random packet arrivals. Ref. \cite{HuangCh14} proves that the threshold-based power control is optimal for single link outage minimization over finite time horizon, and is asymptotic optimal for infinite time horizon. Ref. \cite{Huang13} targets at minimizing the link outage probability with interference from other transmitters, where the author uses random walk theory to find the transmission probability of the EH transmitter. With finite battery capacity, optimal offline water-filling algorithms are developed in \cite{Ozel11,Yener12,RZhang2012}, where the optimal directional water-filling is found based on the causality constraint of using the harvested energy, and suboptimal online policies are also investigated. Our previous work \cite{Gong2013} extends their work to the case with the hybrid transmitter powered by EH and power grid. Ref. \cite{Koksal2013} identifies the optimal power control policy and corresponding performance limit with finite data and energy storages under concave utility functions. Besides the transmission power, some recent studies \cite{Erkip13,Ozel13,Xu14} deal with the circuit power consumption at the transmitter, and ref. \cite{Luo13} addresses the practical half-duplex constraint on the actions of energy harvesting and transmission. To improve the reliability, there have been recent efforts on analyzing packet retransmission mechanisms for EH communication systems\cite{Medepally12}, and ref. \cite{Aprem13} studies an AQR-based retransmission mechanism and applies Markov Decision Processes (MDP) to obtain the optimal power control policy.

However, most of the previous works do not consider the existence of EH receiver, while the energy consumption for receiving cannot be ignored in applications like wireless sensor networks (WSN) with short transmission distances \cite{Cui05,Joseph09,Grover11}. This motivates us to investigate the wireless link consisting of EH transmitter \emph{and} receiver, with circuit blocks, as well as the decoding modules \cite{Cui05} taken into account. It is noted that some works deal with two-hop transmissions with EH source and relays\cite{Thompson12}\cite{Gunduz11}\cite{Huangch13JSAC}, but the energy harvested at the relay is used for the second hop transmission, not for receiving. Among the few, ref. \cite{Joseph09} considers receiving power for multi-hop WSN with EH, and efficient routing and node sleeping algorithms are proposed. With another line of research, energy cooperation with wireless powering is considered \cite{Ulukus13-energcooperation} \cite{Jimin_2013}, where the receiver is powered by the energy transmission from the transmitter. The main difference is that in their cases, the transmitter intentionally delivers the energy on demand to the receiver, while in our case, both sides may have no knowledge of the energy condition of the other.

In this paper, we focus on the outage minimization problem of a fading wireless link with EH transmitter and receiver, under an acknowledgement (ACK) based retransmission mechanism. Only channel distribution information (CDI) is assumed known at the transmitter. Different levels of the battery state information (BSI) sharing between the two nodes are considered. As both nodes are EH-based, it is also nature to investigate the correlations between the energy arrival processes at both sides. Different from \cite{Joseph09}, where static channel is considered and the goal is capacity maximization, we put more focuses on the outage performance of EH links, considering the two-fold randomness from EH and channel fading. The main objective of our research is to find the optimal power control policy, and to investigate how the energy arrival correlation, the battery capacity and different receiver structures can impact the outage performance of a link. The main contributions include:

1) With infinite battery capacity and no retransmission, we prove that threshold-based power control is optimal when the EH process is independent between the transmitter and receiver. The optimality holds regardless of whether BSI is shared or not. We thus provide the closed-form expressions for the optimal power thresholds of the disjoint policy without BSI sharing, and for the joint policy with BSI sharing, respectively. We further investigate the impact of EH correlation between the two nodes, and show that the joint threshold-based policy is still optimal.

2) Based on queuing theory, we provide a simple yet powerful framework of deriving the transmission probability, and the receiving probability of EH links. For the transmission probability, the assumption on the \emph{i.i.d.} EH process in \cite{Huang13} is generalized to stationary and ergodic EH process. The derived transmission and receiving probability further enable directly proving the optimality of the threshold-based policies, which extends the results in \cite{HuangCh14} to the case when both the transmitter and receiver are EH-based, with the consideration of EH process correlation between the two sides, the circuit power, the receiver detection and processing, different levels of BSI and CSI sharing. We also investigate the case when the receiver knows the policy of the transmitter, and show how this information helps improving the system performance.

3) When retransmissions are taken into account, we show that the derived optimal power thresholds for non-retransmission case are still near-optimal. To further adapt the threshold per retransmission, a modified disjoint threshold-based policy called linear power levels policy is adopted, and the choice of the starting power level is investigated.

4) Under finite battery capacity, the proposed policies are analyzed via a three dimensional finite state Markov chain (FSMC) model, which allows the calculation of the optimal power thresholds and evaluation of their performance. The numerical results via the FSMC analysis provide insights to the impact of the system parameters, especially the battery capacity and EH correlation between the transmitter and receiver, on the outage performance of an EH link.

The remainder of this paper is organized as follows. Section \ref{sec:model} sets up the system model, and three power control policies are introduced in Section \ref{sec:pol}. In Section IV, we analyze and prove the optimality of the policies with infinite battery capacity. We investigate our policies with finite battery capacity using FSMC in Section V, and we propose the corresponding local searching algorithm for the optimal thresholds. Numerical results are presented in Section VI, and are followed by our conclusion in Section VII.

\section{System Model and Problem Formulation}
\label{sec:model}
We consider a fading wireless link consisting of one transmitter (denoted by source node $S$) and one receiver (denoted by destination node $D$). Both nodes are powered by batteries of capacity $B_{\max}$ and are capable of harvesting energy from the surrounding environment. Discrete-time model is used, where the time axis is partitioned into slots denoted by $t \in \{0, 1, 2, \cdots \}$. With unit slot length, we will \emph{use energy and power, and their corresponding units interchangeably}.

\subsection{Channel Model}

We consider frequency non-selective block Rayleigh fading channel, where the channel $h$ keeps constant within one time slot, and varies independently from slot to slot, and $\mathbf{E}\{|h|^2\} = 1$. There is only CDI at $S$, while $D$ has CSI. The transmission rate $R$ is assumed fixed. We assume unit bandwidth, and thus $R$ is in the unit of spectrum efficiency. Channel outage occurs when the mutual information between $S$ and $D$ is less than $R$, where we assume one packet per slot. The corresponding channel outage probability $\text{Pr}\left\{ |h|^{2} < \frac{(2^{R}-1)z}{P^t_{\text{tx}}} \right\}$ in any slot $t$ is \cite{Tse05}
\begin{equation}
	p(P^t_{\text{tx}}) =
\left\{
	\begin{aligned}
	& 1 - \exp \left( - \frac{(2^{R}-1)z}{P^t_{\text{tx}}} \right), &P^t_{\text{tx}} > 0,\\
	&1, &P^t_{\text{tx}} = 0,
	\end{aligned}
\right.
\label{eq:outage}
\end{equation}
where $P^t_{\text{tx}}$ is the transmission power at the $S$ in slot $t$, and $z$ denotes the power of additive white Gaussian noise.

\subsection{Energy Consumption Model}
\label{sec:energymodel}
With the consideration of circuit power, the power consumption at $S$ in slot $t$ is given by
\begin{eqnarray}
P_{S}^{t} &=& (1 + \alpha) P_{\text{tx}}^{t} + P_{C, S}\mathbf{1}_{P_{\text{tx}}^{t}>0}, \label{PS}
\end{eqnarray}
where $P_{C, S}$ denotes the static power consumption of the circuit blocks, and $(1 + \alpha) P_{\text{tx}}^{t}$ represents the power consumption for the power amplifier, where $\alpha>0$ mainly relates to the drain efficiency. The indicator function $\mathbf{1}_A$ equals to 1 if event $A$ is true, otherwise it equals to 0. Assume that if $S$ is not transmitting, i.e., $P_{\text{tx}}^{t} = 0$, the circuit does not consume any power, i.e., $P_{S}^{t} = 0$. In practise, transceivers may consume some power even when it is not transmitting or receiving \cite{Cui05}, therefore this assumption implies that $S$ can switch on and off the transceiver instantaneously, or the start-up time is negligible. Similar assumption is made for $D$. Note that at the beginning of each slot, $D$ does not know whether $S$ is transmitting or not in advance. Therefore if there is enough energy in the battery, $D$ will be active, and its power consumption in slot $t$ is given by
\begin{equation}
P_{D}^{t} = P_{D}, \label{PD}
\end{equation}
where $P_{D}$ includes two parts: (1) The power consumption of the receiving circuit, and the decoding module depending on whether the link is coded or uncoded \cite{Grover11}; (2) the power for feeding back ACK to $S$, which will be described in Section \ref{subsect:retransmission}.

To further improve the energy utilization efficiency at $D$, in Section \ref{sec:recv_detect}, we consider that $D$ first spends $\xi$ fraction of the energy $\xi P_D$ to detect whether there is an on-going transmission from $S$, where $\xi \in [0,1]$. In this way, if $D$ is active, the power consumption at $D$ in slot $t$ is
\begin{equation}
P_{D}^{t} =
\left\{
	\begin{aligned}
	\xi P_{D}, &\quad P_S^t = 0,\\
	P_{D}, &\quad P_S^t > 0,
	\end{aligned}
\right. \label{PD-detect}
\end{equation}
where the first case means that $S$ is not transmitting, and thus $D$ will switch off after consuming $\xi P_{D}$ amount of detection energy; otherwise in the second case $D$ will keep receiving for the whole slot when $S$ is transmitting. The power consumption model can be further extended to include data processing power at $D$, which will also be discussed in Section \ref{sec:recv_detect}.

\subsection{Energy Harvesting Model and Battery State Evolution}

Denote the harvested energy in slot $t$ at $S$ and $D$ as $\{E_{S}^{t} \}, \{E_{D}^{t} \} \subset \mathbb{R}_{\ge0}$, which are stationary and \emph{ergodic} sequences over time. Their means are $\mathbf{E} \left[E_{S}^{t} \right] = \lambda_{S}$ and $\mathbf{E} \left[E_{D}^{t} \right] = \lambda_{D}$ respectively. The EH processes can also be correlated between $S$ and $D$, and the correlation coefficient is defined as $\rho = \frac{\text{cov}(E_{S}^{t}, E_{D}^{t})}{\sigma_{E_{S}^{t}} \sigma_{E_{D}^{t}}}$, where $\text{cov}(\cdot)$ is the covariance function, $\sigma_{E_{S}^{t}}$ and $\sigma_{E_{D}^{t}}$ represent the standard variances of $E_{S}^{t}$ and $E_{D}^{t}$.
Let $B_{S}^{t}$ and $B_{D}^{t}$ denote the battery energy at the beginning of slot $t$, of $S$ and $D$ respectively. The battery gets replenished whenever a node harvests energy at the end of the slot. The evolution of the battery states at $S$ and $D$ are
\begin{align}
	B_{S}^{t+1}  & =  \min \left( B_{S}^{t} - P_{S}^{t} + E_{S}^{t}, B_{\max} \right), \label{eq:BSevo}\\	
    B_{D}^{t+1}  & =  \min \left( B_{D}^{t} - P_{D}^t + E_{D}^{t}, B_{\max} \right),
\end{align}
where $B_{\max}$ denotes the battery capacity. Note that the above equation implicitly requires $B_{S}^{t} \geq P_{S}^{t}$ and $B_{D}^{t} \geq P_{D}$ respectively, and these will be enforced by the power control policies, where we remark that $S$ and $D$ know their \emph{own} battery state. Remark that the requirement $B_{D}^{t} \geq P_{D}$ holds even for the case with receiver detection described in (\ref{PD-detect}). This is because at the beginning of each slot, as $D$ does not know whether $S$ is transmitting or not, it has to ensure that $B_{D}^{t} \geq P_{D}$, to guarantee the reception for the whole slot in case $S$ is transmitting.

\subsection{Packet Retransmission}
\label{subsect:retransmission}

In slot $t$, the packet transmission may fail if: (1) $S$ does not have sufficient energy to transmit the packet so that $S$ set $P_{\text{tx}}^t = 0$; (2) The packet is not transmitted successfully due to the channel outage with probability $p(P_{\text{tx}}^t)$; (3) $D$ does not have sufficient energy to receive the packet, i.e., $B_{D}^{t} < P_{D}$. Equivalently, the per-slot transmission successful probability $\phi$ is given by
\begin{equation}
\begin{aligned}
\label{eq:phi}
\phi & = \text{Pr} \{\text{no channel outage} | P_{\text{tx}}^t >0\} \text{Pr} \{ P_{\text{tx}}^t >0, B_D^t \ge P_D \} \\
    & = [1-p(P_{\text{tx}}^t)] \text{Pr} \{ P_{\text{tx}}^t >0, B_D^t \ge P_D \},
\end{aligned}
\end{equation}
where the second equality comes from the independence between the battery states and the channel condition, and note $p(P_{\text{tx}}^t)$ is defined in (\ref{eq:outage}).

The retransmission is based on an ACK-based mechanism, in which ACKs are fed back by $D$ over an error free channel to $S$. If the transmission is successful, $S$ will receive an ACK from $D$ and move on to the next packet. If $S$ does not receive an ACK or it does not transmit due to insufficient battery energy, it retries until it receives an ACK, or reaches the retry limit and moves on to the next packet. We assume: (1) The time for transmitting ACK is negligible; (2) The power consumed for sending ACK is included in $P_{D}$; (3) If channel outage occurs, even without feeding back ACK, the power consumption of $D$ is still $P_{D}$.

As we allow $K - 1$ times of retransmission ($K$ times of transmission attempts in total), the retransmission state is denoted by $u \in \mathcal{U}$, defined by
\begin{equation}
u = \left\{
	\begin{aligned}
	&-1, && \text{ACK received, start a new packet} \\
	&0, && \text{last packet in outage, start a new packet}\\
	&k, && k \text{-th retransmission,} ~ k \in \{1, \cdots, K-1 \}.
	\end{aligned}
\right.
\label{ACK}
\end{equation}
With retransmissions, the final outage probability of a packet $p_{\text{out}}$ is the probability that the packet can not be successfully transmitted within $K$ slots. Please note that with our definition, the ``retransmission" includes the failed transmission attempts due to insufficient energy at $S$.

\subsection{Outage Minimization Problem}
This paper considers the online power control policies at $S$ in order to minimize the packet outage probability $p_{\text{out}}$:
\begin{equation}
\min_{\mathcal{P}} p_{\text{out}},
\end{equation}
where $\mathcal{P}$ denotes the set of all online stationary policies, where in a policy $S$ decides $P_S^t$, based on CDI, average energy arrival rate $\lambda_S$, its BSI (or including $\lambda_D$ and $B_D^t$ at $D$ depending on the level of BSI sharing) and the retransmission state $u$. In this paper we assume the data source of $S$ is backlogged. Our analysis for the backlogged scenario is also valuable for cases with random packet arrivals, in terms of determining the stability condition of the data queue. Specifically, to keep the data queue stable, the bit arrival rate should be less than $R(1-p_{\text{out}})$. Even though the energy buffer and the data queue are coupled at $S$, this stability condition is proved in \cite{Ephremides_2012} by assuming $S$ transmits dummy packets even when the data queue is empty.

\section{Power Control Policies}
\label{sec:pol}
In this section, three power control policies are presented. Their performance will be analyzed in later sections. Note that by default $D$ does not know what policy $S$ is taking, and we will discuss the impact of knowing the policy of $S$ at $D$ in Section \ref{sec:DknowS}.

\subsection{Disjoint Threshold-based Policy}
In this policy, the possible values of $P_{S}^{t}$ are chosen from a binary set and $S$ only transmits when the energy stored in its own battery is larger than a threshold. Specifically,
\begin{equation}
P_S^t =
\left\{
	\begin{aligned}
	P_S, & \quad \text{If } B_S^t \geq P_S, \\
	0, & \quad \text{else},
	\end{aligned}
\right.
\end{equation}
where according to (\ref{PS}), $P_S > P_{C,S}$ is required, and the corresponding transmission power is $P_{\text{tx}}^t = \frac{P_{S} - P_{C, S}}{1 + \alpha}$. The battery state evolution under this policy at $S$ is given by
\begin{equation}
B_{S}^{t+1} = \min ( B_{S}^{t} - P_{S} \mathbf{1}_{B_{S}^{t} \geq P_{S}} + E_{S}^{t}, B_{\max}).
\end{equation}
The optimization variable of the policy is to determine the best threshold $P_S$. Note that in this policy, each node does not know the BSI of the other.

\subsection{Joint Threshold-based Policy}
This is a genie-aided threshold-based policy that each node knows whether the other node will transmit/receive or not in current time slot. Under this policy, $S$ transmits only when both $B_{S}^{t} \geq P_{S}$ and $B_{D}^{t} \geq P_{D}$ hold, and so does $D$ to make receiving decision. The battery state evolution under this policy is
\begin{equation}
\label{eq:battery_joint}
B_{\beta}^{t+1} = \min (B_{\beta}^{t} - P_{\beta}\mathbf{1}_{B_{S}^{t} \geq P_{S}, B_{D}^{t} \geq P_{D}} + E_{\beta}^{t}, B_{\max}),
\end{equation}
where $\beta \in \{S,D\}$. The optimization variable of the policy is the threshold $P_S$.

\subsection{Linear Power Levels Policy}
The previous two policies do not take the retransmission states into account, while this policy is a modified version of the \emph{disjoint} threshold-based policy with non-decreasing multi-level thresholds $\{ P_{S}(u) \}, u \in \mathcal{U}$, i.e.,
\begin{equation}
P_{S}(u) = P_{S}(0) + \Delta \cdot u, \quad u=0,\dots, K-1,
\end{equation}
where $\Delta$ is the step value, and note $P_{S}(-1) = P_{S}(0)$. The battery state evolution at $S$ under this policy is given by
\begin{equation}
B_{S}^{t+1} = \min (B_{S}^{t} - P_{S}(u) \mathbf{1}_{B_{S}^{t} \geq P_{S}(u)} + E_{S}^{t}, B_{\max}).
\end{equation}
The optimization variables of the policy are $P_S(0)$ and $\Delta$.

\section{Infinite Battery Capacity}
\label{sec:infinte}
In this section, we will analyze the outage probability when both nodes are equipped with infinite battery capacity. The optimality of the threshold-based policies is discussed. Unless otherwise specified, the theorems, lemmas, propositions and corollaries have the following default conditions: The battery capacity $B_{\max} = +\infty$. The energy arrivals $\{E_{S}^{t} \}$ and $\{E_{D}^{t} \}$ are stationary and ergodic processes, with $\mathbf{E} \left[E_{S}^{t} \right] = \lambda_{S}$ and $\mathbf{E} \left[E_{D}^{t} \right] = \lambda_{D}$ respectively.

\subsection{Disjoint Threshold-based Policy with $K=1$}
\label{sec:disjoint}
When $K=1$, there is no retransmission. According to the definition in (\ref{eq:phi}), we denote $\phi(P_S)$ as the transmission successful probability in a slot, with threshold $P_S$. The policy design is to find the best threshold $P_S$ so that $\phi(P_S)$ is maximized, written as
\begin{equation}
\label{eq:P1}
\text{(P1)} ~ \max_{P_{S} \geq 0} ~ \phi(P_{S}) \triangleq \exp \left[ - \frac{(2^{R}-1)z}{P_{\text{tx}}} \right] \Psi(P_{S}, P_{D}),
\end{equation}
where recall that $P_{\text{tx}} = (P_S-P_{C,S})/(1+\alpha)$, and $\Psi(P_{S}, P_{D})$ is the probability that $S$ has enough energy to transmit and $D$ has enough energy to receive, i.e.,
\begin{equation}
\Psi(P_{S}, P_{D}) = \lim_{n \to +\infty} \frac{1}{n} \sum_{t = 1}^{n} \mathbf{E} \left[ \mathbf{1}_{B_{S}^{t} \geq P_{S}, B_{D}^{t} \geq P_{D} } \right].
\end{equation}

\begin{figure}[!t]
\centering
\includegraphics[width=2.8in]{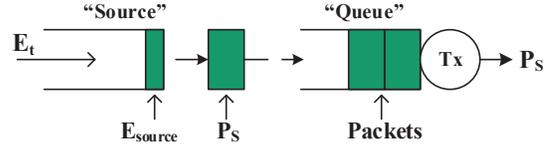}
\caption{The ``energy packet" queue.}
\label{fig:energyQueue}
\end{figure}

To solve (P1), we need to calculate $\Psi(P_{S}, P_{D})$, before which two lemmas are proposed.
\begin{lemma}
\label{lem:transprob}
Denote the transmission probability of $S$ as
\begin{equation}
\Psi_S = \lim_{n \to +\infty} \frac{1}{n} \sum_{t = 1}^{n} \mathbf{E} \left[ \mathbf{1}_{B_{S}^{t} \geq P_{S}} \right].
\end{equation}
Under the disjoint threshold-based policy, we have
\begin{equation}
\label{eq:transprob}
\Psi_S = \min\left(1, \frac{\lambda_S}{P_S}\right).
\end{equation}
\end{lemma}
\begin{proof}
Without loss of generality, assume initially the battery is empty, i.e., $B_S^1=0$. The proof is based on constructing a equivalent energy ``packet" queue $\mathcal{Q}$ reflecting the energy arrival and usage process. The queue operates in a slotted manner. At the end of each slot $t$, the arrived energy is firstly stored in a ``packet source", of which the energy at the beginning of a slot is $E_{\text{source}}^t$. Right at the energy arrival, the ``source" generates energy ``packets", each with volume $P_S$, into $\mathcal{Q}$ at the end the slot. Correspondingly the number of generated packets is $n_t = \lfloor \frac{E_{\text{source}}^t+E_S^t}{P_S}\rfloor$ ($\lfloor \cdot \rfloor$ is the floor operator, and $n_t \in \mathbb{Z}_{\ge0}$), and then the amount of energy
\begin{equation}
\label{eq:Etevo}
E_{\text{source}}^{t+1} = E_{\text{source}}^{t} - n_tP_S+E_S^t
\end{equation}
remains in the ``source", and here we always have $E_{\text{source}}^t < P_S$. Note that $\{n_t\}$ solely evolves with the stationary and ergodic sequence $\{E_S^t\}$, and thus is also stationary and ergodic.
The generated packets are then buffered in $\mathcal{Q}$ with \emph{infinite} capacity. The constructed queue is illustrated in Fig.~\ref{fig:energyQueue}. The ``energy packet" arrival rate of $\mathcal{Q}$, defined as the average arrived packets per slot, is
\begin{eqnarray}
\label{eq:lambda_eq}
\tilde{\lambda}_{\text{ep}} & = & \lim_{n \to +\infty} \frac{\sum_{t=1}^{n}\mathbf{E}\left[n_t\right]}{n} \notag \\
& \overset{\text{(a)}}{=} & \frac{1}{P_S}\lim_{n \to +\infty} \frac{\sum_{t=1}^{n-1}\mathbf{E}[E_S^t]-\mathbf{E}[E_{\text{source}}^{n}]}{n} \notag \\
&\overset{\text{(b)}}{=}&  \frac{\lambda_S}{P_S},
\end{eqnarray}
where (a) holds by adding the two sides of (\ref{eq:Etevo}) from $t=1$ to $t=n$ and then taking expectations, with initial values $E_{\text{source}}^{1} = B_S^1 = 0$; (b) holds because $E_{\text{source}}^{n} < P_S$.

The ``server" of $\mathcal{Q}$  serves one packet per slot if the queue is not empty, i.e., the service rate is $\tilde{\mu}_{\text{ep}}=1$, this corresponds to the policy that $S$ transmits if $B_S^t \geq P_S$. Therefore, denoted by $N_t$ the number of energy packets in $\mathcal{Q}$ at the beginning of time slot $t$, its evolution is
\begin{equation}
N_{t+1} = N_t - \mathbf{1}_{N_t\ge1} + n_t,
\end{equation}
Comparing with (\ref{eq:BSevo}) and noting that we have infinite battery here, it can be easily checked that
\begin{equation}
\label{eq:NtBS}
N_t = \left\lfloor \frac{B_S^t}{P_S} \right\rfloor.
\end{equation}
Therefore,
\begin{eqnarray}
\Psi_S & = & \lim_{n \to +\infty} \frac{1}{n} \sum_{t = 1}^{n} \mathbf{E} \left[ \mathbf{1}_{B_{S}^{t} \geq P_{S}} \right] \notag \\
& = & \lim_{n \to +\infty} \frac{1}{n} \sum_{t = 1}^{n} \mathbf{E} \left[ \mathbf{1}_{N_t \geq 1} \right], \label{eq:little10}
\end{eqnarray}
where the last equality holds due to (\ref{eq:NtBS}). For $\Psi_S$, First consider $\tilde{\lambda}_{\text{ep}} \ge 1$, and we prove $\Psi_S = 1$ by contradiction. Assume $\Psi_S < 1$, and thus for some subsequence of $\{N_t\}$, denote slot sequence $\{t_i\} \subset \mathbb{N}$ as those satisfy $N_{t_i} =0$, and $t_1 < t_2 <\dots$, where the sequence $\{t_i\}$ must be infinite (and thus $\{N_t\}$ is stationary and ergodic), otherwise $\Psi_S = 1$ immediately. Note $\forall t_i >1$, we must have $n_{t_i-1} = 0$, and
\begin{align}
\sum_{t=t_i}^{t=t_{i+1}-2}n_t = t_{i+1}-t_i-1,& \quad \text{if $t_{i+1} > t_i+1$}, \\
n_{t_i} = 0,& \quad \text{if $t_{i+1} = t_i+1$},
\end{align}
where the above equations basically indicate that \emph{between} two adjacent slots with empty queue $\mathcal{Q}$, the number of arrived packets equals to the number of slots between these two slots, because the queue evolves from empty to empty, and each slot exactly consumes one packet. Then
\begin{align}
\lim_{n \to +\infty} \frac{\sum_{t=1}^{n}n_t}{n} &= \lim_{n \to +\infty} \frac{0\cdot\sum_{t=1}^n \mathbf{1}_{N_{t+1}=0}\!+\! \sum_{t=1}^n \mathbf{1}_{N_{t+1}>0}}{n} \notag \\
&= \Psi_S\quad \text{a.s.},
\end{align}
while  $\lim_{n \to +\infty} \frac{\sum_{t=1}^{n}n_t}{n} = \tilde{\lambda}_{\text{ep}}$ almost surely as $\{n_t\}$ is stationary and ergodic, which means $\tilde{\lambda}_{\text{ep}} < 1$ as we assume $\Psi_S < 1$, which contradicts the fact that $\tilde{\lambda}_{\text{ep}} \ge 1$. Therefore $\Psi_S=1$.

When $\tilde{\lambda}_{\text{ep}} < 1$, according to Little's Law \cite{Kleinrock}, $\lim_{n \to +\infty} \frac{1}{n} \sum_{t = 1}^{n} \mathbf{E} \left[ \mathbf{1}_{N_t \geq 1} \right]$ is the server busy probability of $\mathcal{Q}$, and equals to the load of the queue, i.e.,
\begin{equation}
\Psi_S = \frac{\tilde{\lambda}_{\text{ep}}}{\tilde{\mu}_{\text{ep}}} = \frac{\lambda_S}{P_S},
\end{equation}
where the equality holds due to (\ref{eq:lambda_eq}) and $\tilde{\mu}_{\text{ep}}=1$. The proof is completed.
\end{proof}

Note that Lemma \ref{lem:transprob} coincides with \cite[Theorem 1]{Huang13}, while we provide an alternative proof based on queuing theory, and the \emph{i.i.d.} assumption of $\{E_S^t\}$ is relaxed to be stationary and ergodic. This method further allows us to prove other results in this paper with more complicated situations especially for $D$.
\begin{lemma}
\label{lem:recvprob}
Define the receiving probability of $D$ as
\begin{equation}
\Psi_D = \lim_{n \to +\infty} \frac{1}{n} \sum_{t = 1}^{n} \mathbf{E} \left[ \mathbf{1}_{B_{D}^{t} \geq P_{D}} \right].
\end{equation}
Under the disjoint threshold-based policy, we have
\begin{equation}
\label{eq:recvprob}
\Psi_D = \min\left(1, \frac{\lambda_D}{P_D}\right).
\end{equation}
\end{lemma}
\begin{proof}
Follow the same procedure of proving Lemma \ref{lem:transprob}.
\end{proof}

We then have the following Theorem for the simultaneous transmission and receiving probability $\Psi(P_S, P_D)$ when the energy arrivals at $S$ and $D$ are mutually \emph{independent}, i.e., $\rho=0$.

\begin{theorem}
\label{th:Psi}
When EH processes of $S$ and $D$ are mutually independent, and BSI is not shared. We have $\Psi(P_{S}, P_{D})$ under the disjoint threshold-based policy given by
\begin{align}
\Psi(P_{S}, P_{D}) = \min \left( 1, \frac{\lambda_{S}}{P_{S}}, \frac{\lambda_{D}}{P_{D}}, \frac{\lambda_{S}\lambda_{D}}{P_{S} P_{D}} \right).
\end{align}
\end{theorem}
\begin{proof}
Since $\left\{E_{S}^{t} \right\}$ and $\left\{E_{D}^{t}\right\}$ are mutually independent, and also because $S$ and $D$ do not know the BSI of each other, the events $B_S^t\ge P_S$ and $B_D^t \ge P_D$ are independent. We have
\begin{equation}
\Psi(P_{S}, P_{D}) = \Psi_S \Psi_D,
\end{equation}
and according to Lemma \ref{lem:transprob} and Lemma \ref{lem:recvprob}, we directly obtain
\begin{align}
\label{eq:indprob}
\Psi(P_{S}, P_{D}) & = \min\left(1, \frac{\lambda_S}{P_S}\right)\min\left(1, \frac{\lambda_D}{P_D}\right) \notag \\
& = \min \left( 1, \frac{\lambda_{S}}{P_{S}}, \frac{\lambda_{D}}{P_{D}}, \frac{\lambda_{S}\lambda_{D}}{P_{S} P_{D}} \right),
\end{align}
which completes the proof.
\end{proof}

When energy arrivals $\{E_{S}^{t} \}$ and $\{E_{D}^{t} \}$ are not independent, i.e., $\rho\neq0$\footnote{The energy arrival correlation $\rho$ can take negative values under some circumstances, while generally it is positive as for wind and solar energy.}, Theorem \ref{th:Psi} still holds if $\max (\Psi_S, \Psi_D) = 1$. In this case, Theorem \ref{th:Psi} becomes $\Psi(P_{S}, P_{D}) = \min \left(1, \frac{\lambda_{S}}{P_{S}}, \frac{\lambda_{D}}{ P_{D}} \right)$. Otherwise, if $\max (\Psi_S, \Psi_D) < 1$ which means that energy arrivals are insufficient at both nodes, it is expected that $\Psi(P_{S}, P_{D})$ increases with $\rho$, which will be confirmed by our numerical results. But this will highly depend on the energy arrival profiles at both nodes, and no closed-form $\Psi(P_{S}, P_{D})$ is available here.

Next we will show that the disjoint threshold-based policy is actually the \emph{optimal} power control policy under independent energy arrivals. First two lemmas are needed as follows.

\begin{lemma}
\label{lem:outagemin}
Consider a non-EH block fading channel with only CDI at $S$. Given the average power constraint $\mathbf{E}\{P_{S}^t\} \leq P_0$ at $S$, the optimal power control policy that minimizes the average \emph{channel} outage probability, is threshold-based, i.e.,
\begin{equation}\
P_{S}^t = \left\{
	\begin{aligned}
	&0, && \text{with probability $1-\frac{P_0}{P^*}$}, \\
	&P^*, && \text{with probability $\frac{P_0}{P^*}$},
	\end{aligned}
\right.
\end{equation}
where
\begin{equation}\
P^* =\left\{
	\begin{aligned}
	&P_0, && \quad P_0 \geq P_a, \\
	&P_a, && \quad P_0 < P_a,
	\end{aligned}
\right.
\end{equation}
where $P_a$ is the larger solution to the quadratic equation $(P-P_{C,S})^2 = (1+\alpha)(2^R-1)zP$.
\end{lemma}
\begin{proof}
See Appendix \ref{app:lem:outagemin}.
\end{proof}
We then immediately have the following lemma.
\begin{lemma}
\label{lem:outageminEH}
Consider a block fading channel with non-EH $D$ and EH $S$ with stationary and ergodic energy arrivals $\{E_{S}^{t}\}$, and $\mathbf{E} \left[E_{S}^{t} \right] = \lambda_{S}$. Assuming only CDI at $S$, the optimal power control policy that minimizes the average channel outage probability is threshold-based, the same as the corresponding non-EH link with average power consumption at $S$ as $\lambda_S$.
\end{lemma}
\begin{proof}
Let $P_a$ be the same notation as in Lemma \ref{lem:outagemin}. If $\lambda_S \geq P_a$, then $P^* = \lambda_S$. Otherwise, let the policy be transmitting whenever $B_S^t \geq P_a$, according to Lemma \ref{lem:transprob}, the transmission probability is $\Psi_S = \frac{\lambda_S}{P_a}$. This policy is the same as the optimal policy in Lemma \ref{lem:outagemin} by replacing $P_0$ with $\lambda_S$. As EH imposes more strict power allocation causality constraints, the performance of the EH link is no better than the non-EH link with the same average power consumption. Therefore, this policy is optimal for the EH link.
\end{proof}

The above lemma coincides with \cite[Proposition 3.5]{HuangCh14}, while we consider the circuit power $P_{C,S}$ of $S$, and we also provide a direct proof for infinity time horizon without taking the asymptotic to the solution of the finite time horizon case. Now we are ready to prove the optimality of the disjoint threshold-based policy.

\begin{theorem}
\label{th:optimality1}
When EH processes of $S$ and $D$ are mutually independent, and BSI is not shared. Assume no retransmission. Among all stationary power control policies, the disjoint threshold-based policy is \emph{optimal}, with optimal threshold $P_S^*$ as
\begin{equation}
    \label{eq:optimalPS}
	P_{S}^{*} = \max \left(\lambda_{S}, B_{\text{th}} \right),
	\end{equation}
where $B_{\text{th}}$ is given by
\[
B_{\text{th}} = \frac{1}{2} \left[ (2P_{C, S} + c) + c^{\frac{1}{2}}{(4 P_{C, S}+c)}^{\frac{1}{2}} \right],
\]
and $c = (2^R-1)(1+\alpha)z$.
\end{theorem}
\begin{proof}
See Appendix \ref{app:th:optimality1}.
\end{proof}
The theorem also reveals that the optimal threshold does not depend on $\lambda_D$, and thus the optimality holds even when energy arrival statistics are not shared between $S$ and $D$.

\subsection{Joint Threshold-based Policy with K=1}

Under the joint threshold-based policy, $S$ and $D$ know whether the other node will transmit or not, and the battery state evolution is shown in (\ref{eq:battery_joint}). Under this policy, we have the following proposition for the probability $\Psi(P_{S}, P_{D})$.
\begin{proposition}
\label{prop:jointprob}
The probability $\Psi(P_{S}, P_{D})$ under the joint threshold-based policy is
\begin{equation}
\Psi(P_{S}, P_{D}) = \min \left( 1, \frac{\lambda_{S}}{P_{S}}, \frac{\lambda_{D}}{P_{D}} \right).
\end{equation}
\end{proposition}
\begin{proof}
See Appendix \ref{app:prop:jointprob}.
\end{proof}
Using above Proposition, we have the following theorem.

\begin{theorem}
\label{th:optimality2}
Assuming no retransmission, and $S$ and $D$ know whether the other node will transmit or not. Among all stationary power control policies, the joint threshold-based policy is \emph{optimal}. The optimal threshold $P_S^*$ is
	\begin{equation}
\label{eq:optimalPS2}
	P_{S}^{*} = \max \left(\lambda_{S}, B_{\text{th}}, \frac{\lambda_SP_{D}}{\lambda_D} \right),
	\end{equation}
where $B_{\text{th}}$ is the same as in Theorem \ref{th:optimality1}.
\end{theorem}
\begin{proof}
See Appendix \ref{app:th:optimality2}.
\end{proof}

Remark that from Proposition \ref{prop:jointprob} and Theorem \ref{th:optimality2}, the joint threshold-based policy is not affected by the correlation between $\{E_{S}^{t} \}$ and $\{E_{D}^{t} \}$, and its performance and the optimal transmission power threshold are the same \emph{regardless of} the value of $\rho$.

\subsection{Receiver Detection and Processing}
\label{sec:recv_detect}
This section will provide some extensions to the receiver power consumption model. We will first demonstrate a more practical detection structure, and then the energy for processing the received data will be included in the model.

As explained in Section \ref{sec:energymodel}, $D$ can spend $\xi$ portion of $P_D$ to detect whether $S$ is transmitting, rather than receiving for the whole slot. Note that when $\xi=1$, node $D$ has to receive for the whole packet to decide whether $S$ is transmitting or not, which corresponds to the model we discussed previously.
Under disjoint threshold-based policy, we have the following proposition with receiver detection.
\begin{proposition}
\label{prop:recvprob2}
Define the detection and receiving probability of $D$ as
\begin{equation}
\Psi_D = \lim_{n \to +\infty} \frac{1}{n} \sum_{t = 1}^{n} \mathbf{E} \left[ \mathbf{1}_{B_{D}^{t} \geq P_{D}} \right].
\end{equation}
Under the disjoint threshold-based policy, we have
\begin{equation}
\label{eq:probrecvdet}
\Psi_D = \min\left(1, \frac{\lambda_D}{[1-(1-\xi)(1-\Psi_S)]P_D}\right),
\end{equation}
where $\xi$ is the detection cost, and $\Psi_S = \min\left(1,\frac{\lambda_S}{P_S}\right)$.
\end{proposition}
\begin{proof}
See Appendix \ref{app:prop:recvprob2}.
\end{proof}
From (\ref{eq:probrecvdet}), we can see that $\Psi_D \leq \frac{\lambda_D}{P_D}$, which is the receiving probability without detection, and this means that the detection increases the receiving probability. If $\Psi_S < 1$, this gain increases when $\xi$ decreases.
Note that now the events $\mathbf{1}_{B_S^t \geq P_S}$ and $\mathbf{1}_{B_D^t\geq P_D}$ are \emph{no longer independent}, since the energy consumption procedure at $D$ is affected by how energy is used at $S$. Therefore the probability $\Psi(P_S, P_D)$ can not be calculated by (\ref{eq:indprob}). But from simulations, as shown in Fig.~\ref{fig:psitest}, when $\Psi_S <1$ and $\Psi_D<1$, $\Psi(P_S, P_D)$ can be well approximated by $\Psi_S\Psi_D$ for detection only. As a result, following the same procedure as in the proof of Theorem \ref{th:optimality1}, one can still get the threshold $P_S$ that is close to optimal. For joint threshold-based policy, the detection is actually not necessary, or at the optimal operation point, the joint threshold-based policy is equivalent to the case with $\xi=0$. For the linear power levels policy, the optimal performance with receiver detection can be solved by FSMC decribed in the next section.

\begin{figure}[!t]
\centering
\includegraphics[width=3.4in]{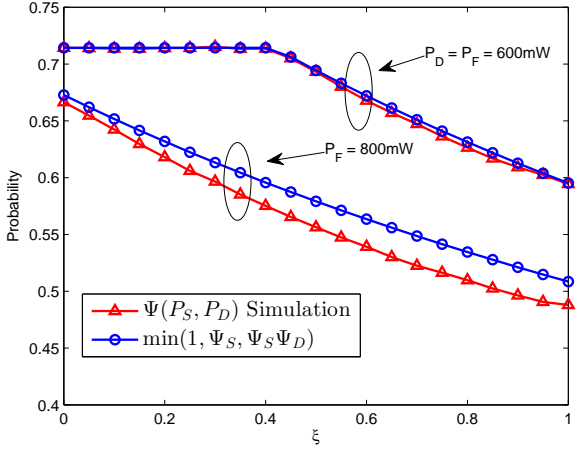}
\caption{Simulation results for $\Psi(P_S,P_D)$ with detection cost $\xi$ and processing, compared with $\Psi(1,\Psi_S,\Psi_S\Psi_D)$, for $P_S = 700$mW, $P_{D} = 600$mW, $z = 100$mW, $\rho = 0$, ${B^{\max} = +\infty}$, and $\lambda_S=\lambda_D = 500$mW. Here $\Psi_S = \frac{\lambda_S}{P_D}$ and $\Psi_D$ is according to (\ref{eq:probrecvdet}) for detection only with $P_D = P_F$, according to (\ref{eq:probrecvdetporoc}) for detection and processing with $P_F >P_D$, respectively.}\label{fig:psitest}
\end{figure}

Moreover, for example in sensor networks, $D$ may use more power to \emph{process} the successfully received data in addition to $P_D$. For instance, $D$ may need to compress and store the sensed data from $S$, or relay the data to the next hop, etc. If the precessing is performed immediately after the data is received, or concurrently with receiving \footnote{Basically it requires that the processing is completed within the same slot.}, we can denote the total power consumption as $P_F$, and correspondingly $P_D = \eta P_F$, where $\eta\in (0,1]$. In this way, at the beginning of a slot, the battery energy should satisfy $B_D^t \ge P_F$, because the battery should support $D$ to receive \emph{and} process the data, otherwise there is no need to receive. The power consumption at $D$ in slot $t$ is
\begin{equation}
P_{D}^{t} =
\left\{
	\begin{aligned}
	\xi \eta P_{F}, &\quad P_S^t = 0,\\
	\eta P_{F}, &\quad P_S^t > 0 \text {~and channel outage occurs}, \\
    P_{F}, & \quad \text{else}.
	\end{aligned}
\right. \label{PD-full}
\end{equation}
where the first case means that $S$ is not transmitting, and thus $D$ will switch off after consuming $\xi P_{D} = \xi \eta P_F$ for detection; In the second case when $S$ is transmitting, $D$ will keep receiving for the whole slot but the channel encounters outage, so that only the receiving power $P_{D} = \eta P_F$ is consumed and no further processing is needed; Finally in the third case, $D$ successfully receives the data and completes the processing.

With receiver detection and processing, the probability $\Psi_D$ is the probability that $B_{D}^{t} \geq P_{F}$. We have the following proposition for $\Psi_D$ with receiver detection and processing.
\begin{proposition}
\label{prop:recvprob3}
Define the detection and receiving probability of $D$ as
\begin{equation}
\Psi_D = \lim_{n \to +\infty} \frac{1}{n} \sum_{t = 1}^{n} \mathbf{E} \left[ \mathbf{1}_{B_{D}^{t} \geq P_{F}} \right].
\end{equation}
Under the disjoint threshold-based policy and with receiver detection and processing, we have
\begin{equation}
\label{eq:probrecvdetporoc}
\Psi_D = \min\left(1, \frac{\lambda_D}{[1\!-\!(1\!-\!\xi\eta)(1\!-\!\Psi_S)\!-\!(1\!-\!\eta)p(P_{\text{tx}})\Psi_S]P_F}\right),
\end{equation}
where $\xi$ and $\eta$ are the detection cost and the receiving fraction respectively, and $\Psi_S = \min \{1, \frac{\lambda_S}{P_S} \}$, and $p(P_{\text{tx}})$ is the channel outage probability defined in (\ref{eq:outage}), with $P_{\text{tx}} = \frac{P_S-P_{C,S}}{1+\alpha}$.
\end{proposition}
\begin{proof}
See Appendix \ref{app:prop:recvprob3}.
\end{proof}
Since the dependence between the events $\mathbf{1}_{B_S^t \geq P_S}$ and $\mathbf{1}_{B_D^t\geq P_D}$ is introduced due to the receiver detection and processing, the probability $\Psi(P_S, P_D)$ can not be calculated by (\ref{eq:indprob}). As shown in Fig.~\ref{fig:psitest}, we can still approximate $\Psi(P_S, P_D)$ with (\ref{eq:indprob}), but the precision becomes worse when we introduce processing, i.e., $P_F > P_D$.

\subsection{Policy of $S$ Known at $D$}
\label{sec:DknowS}
In previous subsections, $D$ has no knowledge about whether $S$ is taking the threshold-based policy or not. But for example, if $D$ knows that $S$ is adopting the threshold-based policy with threshold $P_S$, since $D$ has CSI, when the channel can not meet the transmit rate $R$, $D$ can choose not to receive. In this way, when $B_D^t \ge P_D$,
\begin{equation}
P_{D}^{t} =
\left\{
	\begin{aligned}
	0, &\quad |h|^{2} < \frac{(2^{R}-1)z}{P_{\text{tx}}} ,\\
    P_{D}, & \quad \text{else},
	\end{aligned}
\right. \label{PD-fullpolicyknownatD}
\end{equation}
where $P_{\text{tx}} = \frac{P_S-P_{C,S}}{1+\alpha}$. We have the following results for disjoint/joint threshold-based policies.

\emph{For the disjoint policy}, the following corollary of Theorem 1 holds for the probability $\Psi(P_{S}, P_{D})$.
\begin{corollary}
\label{cor:Psi2}
Assume EH processes of $S$ and $D$ are mutually independent, and BSI is not shared, but $D$ knows that $S$ is taking the threshold-based policy with threshold $P_S$. We have $\Psi(P_{S}, P_{D})$ under the disjoint threshold-based policy given by
\begin{align}
&\Psi(P_{S}, P_{D}) = \notag \\
& \quad \min\left(1, \frac{\lambda_{S}}{P_{S}}, \frac{\lambda_{D}}{P_{D}[1\!-\! p(P_{\text{tx}})]}, \frac{\lambda_{S}\lambda_{D}}{P_{S}P_{D}[1\!-\! p(P_{\text{tx}})]}\right),
\end{align}
where $p(P_{\text{tx}})$ is the channel outage probability defined in (\ref{eq:outage}), with $P_{\text{tx}} = \frac{P_S-P_{C,S}}{1+\alpha}$.
\end{corollary}
\begin{proof}
See Appendix \ref{app:cor:Psi2}.
\end{proof}

Comparing Corollary \ref{cor:Psi2} and Theorem \ref{th:Psi}, we can see that knowing the policy of $S$ at $D$ improves the receiving probability so that the outage performance is improved. We can further prove the optimality of disjoint threshold-based policy as the following corollary of Theorem \ref{th:optimality1}.
\begin{corollary}
\label{cor:optimality1}
Assume EH processes of $S$ and $D$ are mutually independent, and BSI is not shared, but $D$ knows the policy of $S$. Assuming no retransmission, among all stationary power control policies, the disjoint threshold-based policy is \emph{optimal}. The optimal threshold $P_S^*$ is the same as (\ref{eq:optimalPS}) given in Theorem \ref{th:optimality1}.
\end{corollary}
\begin{proof}
See Appendix \ref{app:cor:optimality1}.
\end{proof}

The above argument indicates that, the policy information available at $D$ improves the outage performance, but it has no impact on the optimal policy at $S$, when the BSI is not shared. But the conclusion does not hold when the BSI is shared.

\emph{For the joint policy}, $S$ knows whether $D$ will receive or not, and thus when $D$ knows its policy, which means that the BSI is shared and $S$ also gets the CSI indirectly. As a result, there are full BSI and CSI sharing between $S$ and $D$. We have the following corollary for $\Psi(P_{S}, P_{D})$, and the proof is simply based on the proof of Corollary \ref{cor:Psi2} and Proposition \ref{prop:jointprob}.
\begin{corollary}
\label{cor:jointprob}
Assume $D$ knows that $S$ is taking the threshold-based policy with threshold $P_S$, the probability $\Psi(P_{S}, P_{D})$ under the joint threshold-based policy is
\begin{equation}
\Psi(P_{S}, P_{D}) = \min \left( 1, \frac{\lambda_{S}}{[1- p(P_{\text{tx}})]P_{S}}, \frac{\lambda_{D}}{[1- p(P_{\text{tx}})]P_{D}} \right),
\end{equation}
where $p(P_{\text{tx}})$ is the channel outage probability defined in (\ref{eq:outage}), with $P_{\text{tx}} = \frac{P_S-P_{C,S}}{1+\alpha}$.
\end{corollary}
Based on this corollary, we can easily get the optimal threshold $P_S^*$ for the joint policy when the CSI and BSI are shared.
\begin{corollary}
\label{cor:optimalPSjoint}
When $D$ knows that $S$ is taking the threshold-based policy with threshold $P_S$, the optimal threshold $P_S^*$ of the joint threshold-based policy satisfies
\begin{equation}
\label{eq:optimalPSjoint}
	\exp \left[ - \frac{(2^{R}-1)(1+\alpha)z}{P_S^*-P_{C,S}} \right] = \frac{\lambda_S}{P_S^*},
\end{equation}
where $P_S^* \ge P_{C,S}$.
\end{corollary}

Remark that when both CSI and BSI are shared, the threshold-based policy is no longer optimal, while the truncated channel inversion \cite{Goldsmith05} can be proved optimal, which in our case has no closed-form expression, and is not presented in detail for the conciseness of this paper.

\emph{As for the receiver detection and processing}, first note that in Proposition \ref{prop:recvprob3}, $D$ does \emph{not} know the policy of $S$. If $D$ knows the policy of $S$, Proposition \ref{prop:recvprob3} has the following corollary, of which the proof is simply based on the proof of Corollary \ref{cor:Psi2} and Proposition \ref{prop:recvprob3}.
\begin{corollary}
\label{cor:recvprob3}
Under the disjoint threshold-based policy with receiver detection and processing, if $D$ knows the policy of $S$, we have
\begin{equation}
\label{eq:probrecvdetporoc2}
\Psi_D = \min\left(1, \frac{\lambda_D}{[1-(1-\xi\eta)(1-\Psi_S)][1-p(P_{\text{tx}})]P_F}\right),
\end{equation}
where $\xi$ and $\eta$ are the detection cost and the receiving fraction respectively, and $\Psi_S = \min \{1, \frac{\lambda_S}{P_S} \}$, and $p(P_{\text{tx}})$ is the channel outage probability defined in (\ref{eq:outage}), with $P_{\text{tx}} = \frac{P_S-P_{C,S}}{1+\alpha}$.
\end{corollary}

\subsection{$K>1$ and Linear Power Levels Policy}
\label{subsec:linear}
When retransmission is allowed, i.e., $K>1$, the analysis becomes much more complicated. Mainly because the transmission probability $\Psi$ will be correlated among adjacent slots due to the existence of battery. Even if we constrain the threshold-based policy to be \emph{invariant} among retransmission states, the optimal threshold may be different from our previously derived $P_S^*$. Nevertheless, the derived $P_S^*$ is near-optimal with $K>1$. To illustrate this, consider an equivalent non-EH fading channel with average power constraint $\lambda_S$ and $\lambda_D$ for $S$ and $D$, respectively. Following Lemma \ref{lem:outagemin}, the optimal policy to minimize the per-slot outage is threshold-based. Like the disjoint policy, $D$ has no choice but to receive with probability $\min(1, \frac{\lambda_D}{P_D})$ if it does not know the action at $S$. For this fading channel, threshold $P_S^*$  in (\ref{eq:optimalPS}) is the optimal threshold policy among all policies being \emph{invariant} over retransmissions. The outage probability is then
\begin{equation}
\label{eq:outagefinal}
p_{\text{out}}^{\text{fading}}(P_S^*) = (1-\phi(P_S^*))^K,
\end{equation}
where $\phi(P_S^*)$ is the same as the definition in (\ref{eq:P1}), and the above equation holds since the successful transmissions among adjacent slots are independent for non-EH fading channel. This serves as the lower bound of the outage probability for the counterpart EH link. Therefore, using the same threshold $P_S^*$, the outage probability of the EH link satisfy
\begin{equation}
p_{\text{out}}(P_S^*) \geq p_{\text{out}}^{\text{fading}}(P_S^*) = (1-\phi(P_S^*))^K.
\end{equation}
For joint threshold-based policy, the same conclusion holds. To show that the derived $P_S^*$ is close to optimal, we simulate the outage probability as shown in Fig.~\ref{fig:outagelower}, from which we can see that $p_{\text{out}}(P_S^*)$ is very close to its lower bound $p_{\text{out}}^{\text{fading}}(P_S^*)$, for both joint and disjoint policies.

\begin{figure}[!t]
\centering
\includegraphics[width=3.4in]{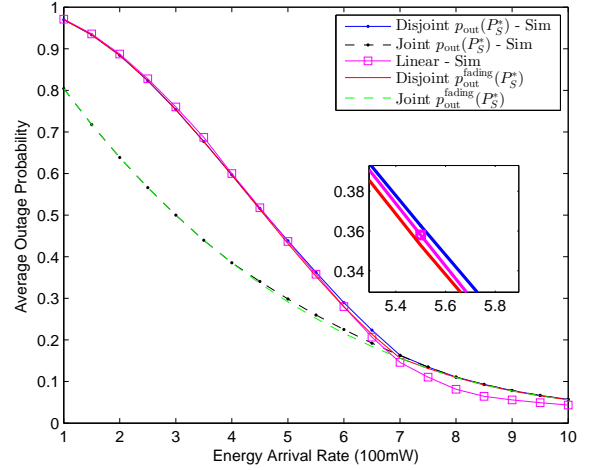}
\caption{Average outage probability versus energy arrival rate for $R = 2$bit/Hz/s, $P_{D} = 700$mW, $z=100$mW, $\rho = 0$, $K=4$, ${B^{\max} = +\infty}$, and other parameters according to Section \ref{sec:num}. Under this setting, $B_{\text{th}} = 787$mW, and $P_S^*$ is determined by (\ref{eq:optimalPS}) and (\ref{eq:optimalPS2}) for disjoint and joint policies, respectively.}
\label{fig:outagelower}
\end{figure}

Different from the other two, the linear power levels policy considers the retransmission state $u$. Given the intervals of the sequence $\Delta$ and the transmission limit $K$, we need to find the optimal starting power level $P_{S,\text{lp}}^{*}$. Since closed-form $p_{\text{out}}$ is not available, by now we can only heuristically choose the starting power. For the counterpart non-EH fading channel, if linear power levels policy is applied, to keep the average power consumption the same, the starting power $P_{S,\text{lp}} \leq P_{S}^{*}$. Using this heuristic, in Fig. \ref{fig:outagelower}, we set $P_{S,\text{lp}} = 0.8P_S^*$ to the EH case, and set $\Delta=100$mW for $\lambda_S<700$mW, and $\Delta=200$mW for $\lambda_S>700$mW. We observe that linear policy outperforms disjoint policy under medium to high energy arrival rates, and its gain increases when energy becomes more sufficient\footnote{Its performance can be better than the non-EH fading channel outage lower bound of joint/disoint policies as these policies are retransmission-invariant.}. This phenomena will also be discussed in Section \ref{sec:num} where we can use FSMC to look for the optimal $P_{S,\text{lp}}^{*}$.

\section{Finite Battery Capacity}
In this section,  we analyze the performance and parameter optimization for the three policies where both nodes are equipped with finite batteries with capacity $B_{\max} < +\infty$. We implement our model using discrete-time FSMC, in which we quantize all energies with respect to a energy fraction $E$ \footnote{We use the unit Joule in our discretized model with the unit slot length, where the energy and power are quantized by the same value $E$ joule.}. Accordingly, the battery levels are considered to be integers multiple of $E$. Here we consider temporal independent Bernoulli distributed energy arrivals at $S$ and $D$, but similar analysis can be extended to other types of distributions, which will be discussed at the end of Section \ref{subsec:FSMC}.

\subsection{FSMC Formulation}
\label{subsec:FSMC}
The discrete-time FSMC has the state space $\mathcal{S} \triangleq \mathcal{B}_{S} \times \mathcal{B}_{D} \times \mathcal{U}$, where $\mathcal{B}_{\beta} \triangleq \{0, 1, \cdots, B_{\max} \}$ ($\beta \in \{S, D\}$) is the set of battery states of node $\beta$. Recall that $\mathcal{U} = \{-1, 0, 1, \cdots, K-1 \}$ is the set of possible packet transmission attempt states. Thus, the state of the link at time $t$ is denoted by $s^{t} \triangleq (b_{S}^{t}, b_{D}^{t}, u^{t})$ where $b_{\beta}^{t} \in \mathcal{B}_{\beta}$ and $u^{t} \in \mathcal{U}$.

For the energy arrival processes, a Bernoulli model $\{E_{S}^{t}, E_{D}^{t} \}$ is considered. At the end of slot $t$, $E_{S}^{\max}$ and $E_{D}^{\max}$ levels of energy are injected into the transmitter and the receiver with probability $\mu_{S}$ and $\mu_{D}$, respectively. With probability $(1 - \mu_{S})$ and $(1 - \mu_{D})$, no energy is harvested. For example, with $\mu_{S} = \mu_{D} = \frac{1}{2}$, their mean values are $\lambda_S = \mathbf{E} [E_{S}^{t}] = \frac{1}{2} E_{S}^{\max}$ and $\lambda_D = \mathbf{E} [E_{D}^{t}] = \frac{1}{2} E_{D}^{\max}$. Similar model is also used in \cite{Aprem13} for wireless sensor networks (WSN). We also consider the correlation between the EH processes of $S$ and $D$ with correlation coefficient $\rho$.

Therefore, the harvested energy pair in slot $t$ can be written as
\begin{equation}
\label{eq:mu}
(E_{S}^{t}, E_{D}^{t}) = \left\{
	\begin{aligned}
	& (0, 0), & ~ \text{with probability } \mu_{0}, \\
	& (0, E_{D}^{\max}), & ~ \text{with probability } \mu_{1}, \\
	& (E_{S}^{\max}, 0), & ~ \text{with probability } \mu_{2}, \\
	& (E_{S}^{\max}, E_{D}^{\max}), & ~ \text{with probability } \mu_{3},
	\end{aligned}
\right.
\end{equation}
where $t = 0, 1, 2, \cdots$. Each of the probabilities $\mu_{i} (i = 0, 1, 2, 3)$ can be calculated depending on $\mu_S$, $\mu_D$ and $\rho$ with the following constraints: (1) $\sum_{i=0}^{i=3} \mu_{i} = 1$ and (2) $\mu_{1} = \mu_{2}$, where the second constraint comes from our setting of changing $E_{\beta}^{\max}$ to control $\lambda_{\beta}$, $\beta \in \{S,D\}$.

Denote the state transition probability matrix as $\mathcal{T}$, whose elements represent the probability of a transition from state $s^{t} \triangleq (m, n, v)$ to state $s^{t+1} \triangleq (i, j, u)$, i.e.,
\begin{align}
\mathcal{T}_{m,n,v}^{i,j,u} = \text{Pr} \left\{ B_{S}^{t+1} = i, B_{D}^{t+1} = j, u^{t+1} = u | \right. \quad\quad \notag \\
 \left.  B_{S}^{t} = m, B_{D}^{t} = n, u^{t} = v \right\},
\label{TransitionMatrix}
\end{align}
where $m, i \in \mathcal{B}_{S}$, $n, j \in \mathcal{B}_{D}$ and $v, u \in \mathcal{U}$. Let $\delta(s^{t}, a^{t}(u), s^{t+1})$ denote the Kronecker delta function\footnote{Without ambiguity with the Dirac delta function used in Appendix \ref{app:lem:outagemin}, as here three parameters are used.} from state $s^{t}$ to $s^{t+1}$ with an action $a^{t}(u) = [P_{S}^{t}(u), P_{D}^{t}]$ taken in slot $t$, with retransmission state $u$, and
\[
\delta(s^{t}, a^{t}(u), s^{t+1}) =  \left\{
	\begin{aligned}
	& 1, & ~ \text{$s^{t}$ transits to $s^{t+1}$ with $a^{t}(u)$}, \\
	& 0, & ~ \text{else},
	\end{aligned}
\right.
\]
For the ease of exposition, in the following we drop the node indexes $S$ and $D$ to consider one of the two nodes, e.g., we use $\mu$ to denote the EH probability at each slot. For $\forall v$ and $u = -1$, $\mathcal{T}_{s^{t}}^{s^{t+1}} = \mu \left(1 - p(P_{\rm tx}) \right) \delta(s^{t}, a^{t}(u), s^{t+1})$ if a transmission successes. For $v = K-1$ and $u = 0$, or $v = -1$ and $u = 1$, or $v \ne -1, K-1$ and $u = v+1$, $\mathcal{T}_{s^{t}}^{s^{t+1}} = \mu p(P_{\rm tx}) \delta \left(s^{t}, a^{t}(u), s^{t+1} \right)$ if a transmission fails because an outage occurs. For $v = K-1$ and $u = 0$, or $v = -1$ and $u = 1$, or $v \ne -1, K-1$ and $u = v+1$, $\mathcal{T}_{s^{t}}^{s^{t+1}} = \mu \delta \left(s^{t}, a^{t}(u), s^{t+1} \right)$ if a transmission fails due to insufficient energy states. Note in above explanations, the probability $\Psi(P_S^t,P_D^t)$ is reflected in $ \delta(s^{t}, a^{t}(u), s^{t+1})$, where the action $a^{t}(u) = \left[P_{S}^{t}(u), P_{D}^{t} \right]$ varies according to different power control policies as proposed previously. Specifically, for disjoint/joint threshold-based policy, $P_S^t$ does not depend on $u$. For $D$, when the receiver detection is considered, $P_D^t$ is decided by (\ref{PD-detect}). When $D$ knows the policy of $S$,  $P_D^t$ is decided by (\ref{PD-fullpolicyknownatD}), and $P_D^t$ takes zero value with probability $p(P_{\text{tx}})$ and this must coincide with the calculation of $\mathcal{T}_{s^{t}}^{s^{t+1}}$. When both detection and processing are considered, $P_D^t$ is decided by (\ref{PD-full}) with channel outage probability $p(P_{\text{tx}})$, which should also coincide with the calculation of $\mathcal{T}_{s^{t}}^{s^{t+1}}$. Note that $\xi = 1$ and $\eta=1$ represent the cases without receiver detection or processing.

The stationary probabilities $\pi(s)$ of this FSMC can be obtained by solving the balance equation\footnote{Under mild conditions on the transmission power policy, we can show that the FSMC is irreducible and positive recurrent, which is not detailed here.},
\begin{equation}
\pi(s) = \sum_{s} \text{Pr} \left\{s^{t+1} = (i, j, u) | s^{t} = (m, n, v) \right\} \pi(s).
\label{Pi}
\end{equation}
Under the ACK-based retransmission mechanism, an outage occurs if and only if a packet is still not transmitted successfully even after $K-1$ times of retransmission. We focus on the following problem of minimizing the average outage probability
\begin{equation}
\min_{P_{S} \geq P_{C,S}} p_{\text{out}},
\end{equation}
where the optimization variable $P_S$ is the power threshold for the disjoint/joint threshold-based policies and is the power levels threshold for the first attempt in the linear power policy, and the step $\Delta$ is given. The average outage probability $p_{\text{out}}$ is decided in the following proposition.

\begin{proposition}
The average outage probability is given by
	\begin{equation}
	p_{\text{out}} = \frac{\pi_{u=0}(s)}{\pi_{u=0}(s) + \pi_{u = -1}(s)},
	\label{OP}
	\end{equation}
where $\pi(s)$ is the stationary probability that the link is in state $s = (b_{S}, b_{D}, u)$, and $\pi_{u=k}(s)$ is defined as
\[
\pi_{u=k}(s) = \sum_{b_{S},b_{D}}\pi(b_{S}, b_{D}, u=k).
\]
\end{proposition}
\begin{proof}
Notice that $\pi_{u=0}(s)$ represents the probability that a transmission outage occurs in the previous slot, and $\pi_{u=-1}(s)$ represents the probability that a successful transmission is finished in the previous slot. Otherwise, in the states $s$ where $u \ne -1, 0$, the packet will be under retransmission in the current slot. Assuming we have $N$ slots and $N \to +\infty$ to reach the stationary distribution, then the number of packets transmitted successfully is $N_{\text{success}} = N \pi_{u = -1}(s)$, and the number of  packets transmitted with an outage is $N_{\text{outage}} = N \pi_{u = 0}(s)$. Thus, the average outage probability is given by the ratio of the packets that are not transmitted successfully even after $K-1$ times of retransmission:
	\begin{equation}
	p_{\text{out}} = \frac{N_{\text{outage}}}{N_{\text{outage}} + N_{\text{success}}},
	\end{equation}
and we get the proposition.
\end{proof}

Similarly, the average transmission times $\tau$ is given by a weighed function of the stationary probability $\pi(s)$. First define,
\[
p_{u,v} = \sum_{(i,j)(m,n)} \text{Pr} \left\{s^{t+1} = (i, j, u) | s^{t} = (m, n, v) \right\},
\]
and then
\begin{equation}
\tau = 1+ \frac{\sum_{k = 1}^{K-1} k p_{-1,k} \pi_{v=k}(s)}{\pi_{u=-1}(s)}.
\label{ReTimes}
\end{equation}

We can now summarize the procedure for computing the outage probability $p_{\text{out}}$ and the average transmission times $\tau$:

\begin{enumerate}
\item Discretize the battery levels and the power consumption levels to generate the FSMC with state space $\mathcal{S}$,
\item Compute $\mathcal{T}$, where $\mathcal{T}$ is the state transition probability matrix with entries defined in (\ref{TransitionMatrix}),
\item Obtain the stationary probabilities $\pi(s)$ by solving (\ref{Pi}),
\item Obtain $P_{\text{out}}$ and $\tau$ from (\ref{OP}) and (\ref{ReTimes}), respectively.

\end{enumerate}

\begin{remark}
In the infinite battery capacity case where ${B_{\max} = +\infty}$, it is obtained that if $\lambda_{S} \geq P_{S}$ and $\lambda_{D} \geq P_{D}$, the energy is always enough, and the battery states will have the behaviors ${B_{\beta} = +\infty}$ ($\beta \in \{S, D\}$). Thus the state space will become a simpler set denoted as $\mathcal{S}_{U} = \mathcal{U} = \{-1, 0, 1, \cdots K-1 \}$. The corresponding transition matrix $\mathcal{T}_{U}$ is of dimension $(K + 1)$ and the stationary distribution can be easily obtained as
\begin{equation}
\pi_{U} = \frac{1-p}{1-p^{K}} [1, p^{K}, p, p^{2}, \cdots, p^{K-1}],
\label{FSMC_Remark}
\end{equation}
where $p = p(P_{\rm tx})$. Through (\ref{FSMC_Remark}), $p_{\text{out}} = p^{K}$ (which means $K$ times of transmission fail continuously) where $p$ is obtained through (4). However, for the general value of $B_{\max}$, it is hard to get the closed-form of the stationary distribution $\pi(s)$.
\end{remark}

\begin{remark}
Although our analytical results in Section \ref{sec:infinte} hold for general stationary and \emph{ergodic} EH process, we restrict our FSMC analysis to \emph{i.i.d.} EH sequences. Because the main focus is on the correlation between the EH processes of $S$ and $D$. According to the results in Section \ref{sec:infinte}, the temporal correlation of the EH process will not have notable impact on the outage performance when the battery capacity is large. Nevertheless, temporal correlation may have severe impact on the outage performance with small battery. In fact, our FSMC formulation can be extended to tackle the temporal correlated EH process. Specifically, the current system state $\mathcal{S}$ has to be expanded to be five-dimensional, including $E_S^t$ and $E_D^t$ as additional state attributes. This however imposes high calculation complexity, which should be addressed in future studies.
\end{remark}

\subsection{Local Searching Algorithm}
For the disjoint threshold-based policy and the joint threshold-based policy requires carefully choosing the optimal threshold $P_{S}^{*} \in \mathcal{B}_{S}$. For the linear power levels policy, one also needs to find the optimal starting point $P_{S,\text{lp}}^{*} = P_{S}({-1}) = P_{S}({0})$, for given $K$ and the interval $\Delta$ \footnote{In the algorithm, note that $E$ represents the discretized energy unit.}. Since it is not possible to directly get the closed-form solution of $p_{\text{out}}$, a one-dimension local searching algorithm shown in Algorithm 1 is implemented to find the optimal $P_{S}^{*}$ and the corresponding $P_{\rm tx}^{*}$. This algorithm requires searching all the states of the energy storage at the transmitter and is of complexity $\mathcal{O}(B_{\max}+1)$. We also have the transition matrix $\mathcal{T} \in \mathbb{R}^{(B_{\max}+1)^{2} \times (K+1)}$, and the detailed $\mathcal{T}$ will be different according to different power control policies.

\begin{algorithm}[t]
\caption{One Dimensional Searching Algorithm}
\begin{algorithmic}[1]
\STATE {\bf Initialize} $P_{S}^{*} = 0$ and $p_{\text{out}} = 1$;
\FOR{$P_{S} = P_{C,S}, P_{S} \in \mathcal{B}_{S}$}
\STATE Calculate the transition matrix $\mathcal{T}$ via (\ref{TransitionMatrix}) and determine the stationary distribution $\pi$ via (\ref{Pi});
\STATE Obtain the outage probability $\hat{p}_{\text{out}}$ via (\ref{OP});
\IF{$\hat{p}_{\text{out}} \leq p_{\text{out}}$}
\STATE $p_{\text{out}} \leftarrow \hat{p}_{\text{out}}$;
\STATE $P_{S}^{*} \leftarrow P_{S}$;
\STATE $P_{\rm tx}^{*} = \frac{P_{S}^{*} - P_{C, S}}{1 + \alpha}$;
\ENDIF
\STATE $P_{S} = P_{S} + E$;
\ENDFOR
\end{algorithmic}
\end{algorithm}

\section{Numerical Results}
\label{sec:num}
In this section, numerical results are presented to show the performance of the proposed power control policies. To use FSMC, the state space is discretized with the granularity of $E = 50$mW, and thus the energy and power are of the same quantization granularity. We consider temporal independent Bernoulli energy arrival with parameters $E_{\max} = 20 E$ and $\lambda_{S}=\lambda_D=500$mW$= 10E$, and $P_{C, S} = 100$mW$= 2E$, $P_{D} = 700$mW$=14E$ corresponding to the coded system \cite{Cui05}\cite{Grover11}. The noise power $z=100$mW. We also set $P_F=P_D$, i.e., $\eta=1$. Given the correlation coefficient $\rho$, we can obtain $\mu_{i} (i = 0, 1, 2, 3)$ from (\ref{eq:mu}). Parameter $\alpha = 1$ corresponding to the drain efficiency of $0.5$ for a typical class-AB amplifier \cite{Berglund06}, and thus $P_{S}^{t} = 2P_{\text{tx}}^{t} + P_{C, S}$. For the retransmission mechanism, unless otherwise mentioned, we allow $K = 4$ times of transmission ($3$ times of retransmission), and thus the retransmission state set $\mathcal{U} = \{-1, 0, 1, 2, 3 \}$. In all sets of results, the battery capacity is finite, so that the optimal power control policy parameters, i.e., $P_S^*$ for joint/disjoint threshold-based policy, $P_{S,\text{lp}}^{*}$ for linear power levels policy (with $\Delta = 100$mW) are decided by the FSMC-based local search in Algorithm 1. Remark that the linear policy under evaluation in this section is with $\xi=1$. To evaluate the gain provided by knowing the policy of $S$ at $D$, we provide corresponding results for joint threshold-based policies. In other policies, this setting is default, i.e., $D$ does not know the policy of $S$.

\begin{figure}[!t]
\centering
\includegraphics[width=3.8in]{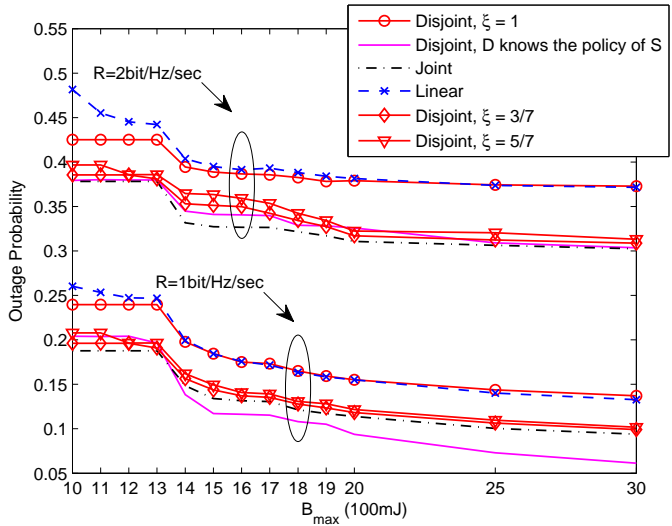}
\caption{Average outage probability with varying $B_{\max}$ and detection cost $\xi$, for $R \in \{1, 2 \}$bit/Hz/s, $\rho=0.5$, $K=4$.} \label{fig:Bmax}
\end{figure}

First in Fig.~\ref{fig:Bmax}, the impact of battery capacity $B_{\max}$ is observed. For all policies, the outage performance degrades when $B_{\max}$ decreases. Because of the energy overflow with finite battery, the effective energy that can be used is reduced, resulting in more outages due to power shortage. With receiver detection, the energy at $D$ is used more efficiently, and thus we observe less outage with smaller value of detection cost $\xi$. It is also shown that the linear policy is worse than the disjoint policy with small battery capacity. As the linear policy requires larger transmission power for retransmissions, small battery cannot provide enough energy, so that the transmission probability $P_S(u)$ for retransmissions $u > 1$ is reduced, which degrades the performance. While for large battery, linear policy shows performance gain over disjoint policy, which coincides with our results for infinite battery in Fig.~\ref{fig:outagelower}.

\begin{figure}[!t]
\centering
\includegraphics[width=3.8in]{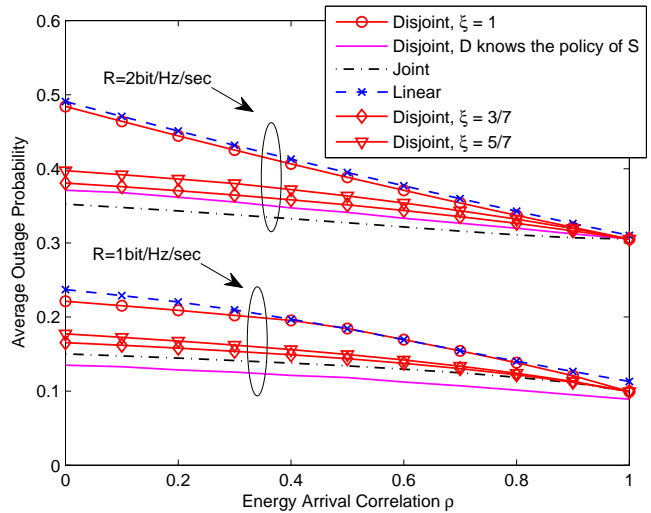}
\caption{Average outage probability with varying $\rho$ and detection cost $\xi$, for $R \in \{1, 2 \}$bit/Hz/s, $B_{\max}=3\lambda_S = 1500$mJ, $K=4$.} \label{fig:rho}
\end{figure}

We then show the results of varying energy arrival correlation $\rho$ in Fig.~\ref{fig:rho}. It is observed that energy arrival correlation helps improve the outage performance. Because correlation has the effect of ``matching" the energy states of $S$ and $D$, which can increase the probability that both nodes have enough energy to transmit and receive in the same slot. With infinite battery, as shown in Proposition \ref{prop:jointprob} and Theorem \ref{th:optimality2}, the performance of the joint policy is irrelevant to $\rho$, but under finite battery, its performance still improves with increasing $\rho$ \emph{slightly}. As the battery capacity is only $B_{\max} = 3\lambda_S$, linear power levels policy performs the worst.

In Fig.~\ref{fig:Bmax} and Fig.\ref{fig:rho}, we also provide curves for the performance of disjoint policy when $D$ knows that $S$ is taking the threshold-based policy. It turns out that knowing the policy of $S$ at $D$ lets disjoint policy perform close to joint policy, and even better than joint policy when $R$ is small, and they are both insensitive to $\rho$. The reason of their similar performance is that if $D$ knows the threshold of $S$, its energy is saved by avoiding unnecessary receiving when channel undergoes deep fading, while joint policy saves the energy of $D$ by avoiding unnecessary receiving when $S$ is not transmitting. We should remark that which of them performs better should depend on system parameters including channel statistics, rate requirement and EH profiles.

\begin{figure}[!t]
\centering
\includegraphics[width=3.8in]{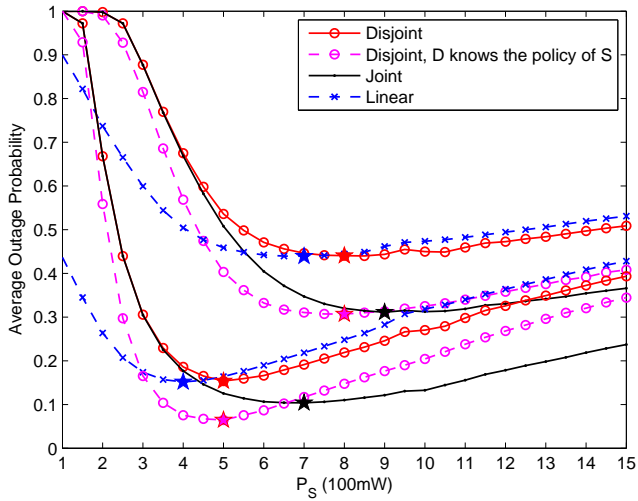}
\caption{Average outage probability with varying $P_S$ for disjoint/joint threshold-based policies and $P_S(0)$ for linear power levels policy, for $R \in \{1, 2 \}$bit/Hz/s, $B_{\max}=6\lambda_S = 3000$mJ, $\xi=1$, and $\rho = 0$, $K=4$.} \label{fig:findPS}
\end{figure}

To illustrate the procedure of searching for optimal threshold $P_S^*$ with Algorithm 1, in Fig.~\ref{fig:findPS}, the outage performance with different values of $P_S$ is shown. The optimal $P_S^*$ for each policy is labeled with a star symbol. Since $\rho=0$ and the battery capacity is relatively large, i.e., $B_{\max}=6\lambda_S$, the optimal thresholds are compared with those derived under infinite battery and no retransmission. For infinite battery with $R=2$bit/Hz/s, according to Theorem \ref{th:optimality1}, $P_S^* = 787$mW for disjoint and joint threshold-based policies. While in Fig.~\ref{fig:findPS}, due to the quantization granularity and the existence of retransmission, the optimal $P_S^*$ is $800$mW and $900$mW for these two policies respectively. The case of disjoint policy is very close to our theoretical result, while it seems that joint policy is more sensitive to the existence of retransmission. With $R=1$bit/Hz/s, similar observations are made. Please also note that the optimal starting power of the linear policy is smaller than $P_S^*$ of the disjoint policy, which confirms our intuition in Section \ref{subsec:linear}. With large battery, the linear policy shows better performance than the disjoint policy, and the gain is more evident with smaller $R$. Fig.~\ref{fig:findPS} also shows the results for disjoint policy when $D$ knows the policy of $S$. Although the optimal outage probability is better than the default disjoint policy, the optimal threshold is the same, which coincides with Corollary \ref{cor:optimality1}.

\begin{figure}[!t]
\centering
\includegraphics[width=3.7in]{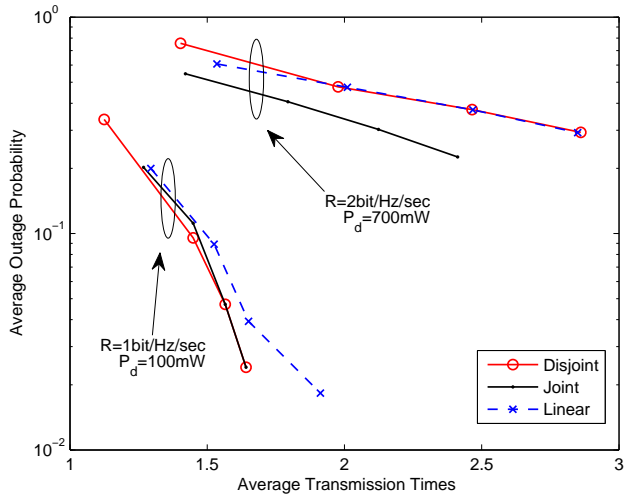}
\caption{The tradeoff between the minimum $p_{\text{out}}$ and the corresponding average transmission times $\tau$, with different maximum transmission times $K\in\{2,3,4,5\}$, for $R \in \{1, 2 \}$bit/Hz/s, $B_{\max}=6\lambda_S = 3000$mJ, $\xi=1$, $\rho=0.5$.}
\label{Fig4_1}
\end{figure}

Finally, we show the tradeoff between the average outage probability $p_{\text{out}}$ and the average transmission times $\tau$ in Fig. \ref{Fig4_1}. We change the number of maximum transmission times $K\in\{2,3,4,5\}$ and get different pairs of $(\tau, p_{\text{out}})$, i.e., the four points on each curve. When we set higher $K$, $p_{\text{out}}$ goes down. For $K=2$, the linear policy shows better performance over the disjoint policy. We also try a smaller value of $P_D = 100$mW, corresponding to uncoded system \cite{Cui05}. In this case, the smaller $P_D$, together with the smaller $R=1$bit/Hz/s, provides relatively sufficient energy availability, the linear policy can achieve the best outage probability, with the cost of larger delay.

\section{Conclusion}
We have investigated the outage minimization problem for a fading wireless link with EH transmitter and receiver. Three power control policies are proposed, namely the disjoint threshold-based policy, joint threshold-based policy, and the linear power levels policy. For infinite battery and independent energy arrivals, we have proved that the two threshold-based policies are optimal among all retransmission invariant power control policies, with or without BSI sharing between $S$ and $D$, respectively. Specifically for the joint threshold-based policy, the optimal threshold and its performance are independent of the energy arrival correlation $\rho$ between $S$ and $D$. We also consider the receiver detection, receiver processing, and different levels of BSI and policy information sharing in our analysis. To analyze the performance with finite battery capacity, we formulate the EH link into an FSMC, with which the optimal power thresholds can be calculated. With extensive numerical tests via FSMC, we have shown the performance of the three policies under different values of the battery capacity $B_{\max}$ and the energy arrival correlation $\rho$. It is found that larger correlation $\rho$ actually improves the outage performance. We also show that knowing the policy of $S$ at $D$ can enhance the performance of disjoint policy to be comparable to joint policy. Unlike the infinite battery case, the linear policy shows its performance gain over disjoint policy only with larger battery capacity or lower $R$. Finally, the tradeoff between the average outage probability and average retransmission times is presented.

Possible extensions of this work are as follows. First, in this paper we consider \emph{i.i.d.} block fading channel, while for Markovian channel, $S$ can exploit the ACKs to estimate the channel gain and even learn the EH process of $D$. MDP based schemes can be promising candidates. Moreover, when the power consumption in the idle state (when the transceiver is not transmitting or receiving) and the switch-on/off delay are non-negligible, an appropriate node sleeping policy \cite{Joseph09} should be designed, and can be jointly optimized with the power control policy.

\appendices

\section{Proof of Lemma \ref{lem:outagemin}}
\label{app:lem:outagemin}
Denote the power control policy as a probability dense function (PDF) $g(P)$, since $S$ only has CDI. From (\ref{eq:outage}), write the per-slot transmission \emph{successful} probability as the function of the total power consumption $P$:  $f(P) = \text{e}^{-\frac{(1+\alpha)(2^R-1)z}{P-P_{C,S}}}$ for $P > P_{C,S}$, otherwise $f(P) = 0$. The objective is to find the optimal $g(P)$
\begin{equation}
\label{eq:choutagemin}
\max_{g(P)} \int_0^{+\infty} g(P) f(P) \text{d} P,
\end{equation}
subject to $\int_0^{+\infty} Pg(P) \text{d} P = P_0$. It is easy to check that $f(P)$ is convex over $[0 , P_{b}]$ and concave over $[P_{b}, +\infty)$, where $P_{b} = \frac{(1+\alpha)(2^{R}-1)}{2} + P_{C,S}$.
Furthermore, by solving $\frac{f(P)}{P} = f'(P)$\footnote{For $P\ge P_{C,S}$, as $f(P)$ is differentiable on $(P_{C,S},+\infty)$.}, we get $P_a$ \footnote{The larger solution to the quadratic equation $(P-P_{C,S})^2 = (1+\alpha)(2^R-1)zP$.}. One can check that $P_a > P_b$, and geometrically the line from $(0,0)$ to $(P_a, f(P_a))$ is a tangent to $f(P)$. We thus construct a function
\begin{equation}
q(P) = \left\{
	\begin{aligned}
	\frac{f(P_a)P}{P_a}, &\quad P \in [0,P_a], \\
	f(P), & \quad P \in [P_a, +\infty],
	\end{aligned}
\right.
\end{equation}
which is a continuous and concave function as $P_a > P_b$. The functions and points used in the proof are illustrated in Fig.~\ref{fig:lemma3proof}. Since $f(P) \leq q(P)$, then $\forall g(P)$
\begin{equation}
\label{eq:g_q_ineq}
\int_0^{+\infty} g(P) f(P) \text{d} P \leq \int_0^{+\infty} g(P) q(P) \text{d} P.
\end{equation}
Denote
\begin{equation}
\label{eq:def_q}
\hat{g}(P) = \arg \max_{g(P)} \int_0^{+\infty} g(P) q(P) \text{d} P,
\end{equation}
subject to $\int_0^{+\infty} Pg(P) \text{d} P = P_0$. Since $q(P)$ is concave, for the random variable $P$ with PDF $g(P)$, we have \cite[(3.5)]{Boyd}
\[
 \mathbf{E}\{q(P)\} \leq q(\mathbf{E}\{P\}),
\]
i.e.,
\begin{equation}
\int_0^{+\infty} g(P) q(P) \text{d} P \leq q(P_0),
\end{equation}
and the equality holds when we set $\hat{g}(P) = \delta(P_0)$, where $\delta(x)$ is the Dirac delta function, which means that transmitter always transmits with $P_0$ power.
Let
\begin{equation}
\label{eq:optimalgp}
g^*(P) = \left\{
	\begin{aligned}
	\delta(P_0), &\quad P_0 \geq P_a, \\
	(1- \frac{P_0}{P_a})\delta(0) + \frac{P_0}{P_a}\delta(P_a), &\quad P_0 < P_a.
	\end{aligned}
\right.
\end{equation}
Note $g^*(P)$ satisfies $\int_0^{+\infty} Pg^*(P) \text{d} P = P_0$. We then have
\begin{align}
\int_0^{+\infty} g^*(P) f(P) \text{d} P & \overset{(a)}{=}  \int_0^{+\infty} g^*(P) q(P) \text{d} P \notag \\
& \overset{(b)}{=}  \int_0^{+\infty} \hat{g}(P) q(P) \text{d} P,
\end{align}
where $(a)$ holds due to the definition of $q(P)$ and $g^*(P)$, and $(b)$ holds since $q(P)$ is linear in $[0, P_a]$.
Therefore, according to (\ref{eq:g_q_ineq}) and (\ref{eq:def_q}), $g^*(P)$ is the optimal solution to (\ref{eq:choutagemin}).

\begin{figure}[!t]
\centering
\includegraphics[width=3.5in]{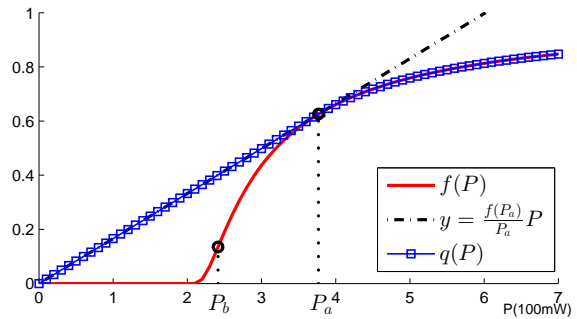}
\caption{In the illustrative figure we set $R=0.5$bit/Hz/s, $\alpha=1$, $P_{C,S}=200$mW, $z = 100$mW.}
\label{fig:lemma3proof}
\end{figure}

\section{Proof of Theorem \ref{th:optimality1} }
\label{app:th:optimality1}
For policies that do not depend on $u$, the objective is to maximize the per-slot successful probability $\phi$. Since energy arrivals $\{E_{S}^{t} \}$ and $\{E_{D}^{t} \}$ are independent, multiplying $\Psi_D$ with $f(P)$ (defined in Appendix \ref{app:lem:outagemin}) we get the successful probability of a transmission from $S$ as $\hat{f}(P) = \Psi_D f(P)$. As $\Psi_D$ is constant, $\hat{f}(P)$ keeps the properties of $f(P)$ used in the proof of Lemma \ref{lem:outagemin} and \ref{lem:outageminEH},  and thus the optimal policy at $S$ is threshold-based.

To find the optimal threshold $P_S^*$, we need to solve (P1) in (\ref{eq:P1}), where $\Psi(P_S, P_D)$ is given in Theorem \ref{th:Psi}. Note $c \triangleq (2^{R}-1)(1+\alpha)z$, $P_{D}$, $\lambda_{S}$ and $\lambda_{D}$ are all positive constants. We study the following cases.

Case 1: $\lambda_{S} \geq P_{S}$. In this case, $\Psi(P_{S}, P_{D}) = \min \left(1, \frac{\lambda_{D}}{P_{D}} \right)$ is independent of $P_{S}$. Thus $\phi(P_{S}) = \text{e}^{\left( - \frac{c}{P_{S} - P_{C, S}} \right)} \Psi(P_{S}, P_{D})$, and by obtaining the first-order derivative we have
\[
\frac{\text{d} \phi (P_{S}) }{\text{d} P_{S}} = \min \left(1, \frac{\lambda_{D}}{P_{D}} \right) \text{e}^{\left( - \frac{c}{P_{S} - P_{C, S}} \right)} \frac{c}{(P_{S} - P_{S, C})^{2}} > 0
\]
holds for $\forall P_{S}$. This further yields that $\phi(P_{S})$ is an increasing function in its domain $P_{S} \in [0, \lambda_{S}) \backslash \{P_{C, S} \}$. Obviously, $P_{S} > P_{S, C}$ should hold.

Case 2: $\lambda_{S} < P_{S}$. In this case, $\Psi(P_{S}, P_{D}) = \min \left(1, \frac{\lambda_{D}}{P_{D}} \right) \frac{\lambda_{S}}{P_{S}}$. Solving $\frac{\text{d} \phi (P_{S})}{\text{d} P_{S}} = 0$ leads to
\begin{equation}
c P_{S} = (P_{S}-P_{C, S})^{2}.
\end{equation}
By solving the above equation we have
\begin{equation}
B_{\text{th}}^{1, 2} = \frac{1}{2} \left[ (2P_{C, S} + c) \pm c^{\frac{1}{2}}{(4 P_{C, S}+c)}^{\frac{1}{2}} \right],
\label{Roots}
\end{equation}
and $B_{\text{th}}^{1} < B_{\text{th}}^{2}$, and we can check that $B_{\text{th}}^{1} < P_{C,S}$, which is not reasonable. So we define $B_{\text{th}} \triangleq B_{\text{th}}^{2}$, and it is guaranteed that $B_{\text{th}}^{2} > P_{C,S}$. If $\lambda_{S} < B_{\text{th}} $, then $\left. \frac{\text{d}^{2} \phi (P_{S}) }{\text{d} P_{S}^{2}}  \right|_{P_{S} = B_{\text{th}}} < 0$. Thus $P_{S}^{*} = B_{\text{th}}$ if $\lambda_{S} < B_{\text{th}}$. Otherwise, if $\lambda_{S} \geq B_{\text{th}}$, $\phi(P_{S})$ decreases in $[\lambda_{S}, +\infty)$, and hence $P_{S}^{*} = \lambda_{S}$.

From the two cases above, we finally have $P_{S}^{*} = \max \left(\lambda_{S}, B_{\text{th}} \right)$. The proof is completed.

\section{Proof of Proposition \ref{prop:jointprob}}
\label{app:prop:jointprob}
First, if $\max(\frac{\lambda_S}{P_S}, \frac{\lambda_D}{P_D}) \geq 1$, then following the proof of Lemma \ref{lem:transprob}, one or both of the two nodes always have enough energy to transmit or receive. The problem degenerates to calculating the one side transmission or receiving probability. Then according to Lemma \ref{lem:transprob} and \ref{lem:recvprob}, we have $\Psi(P_{S}, P_{D}) = \min(\Psi_S, \Psi_D) = \min(1, \frac{\lambda_S}{P_D}, \frac{\lambda_S}{P_D})$.

If $\max(\frac{\lambda_S}{P_S}, \frac{\lambda_D}{P_D}) < 1$, due to the constraint that the used energy should be less than the arrived energy, we have
\[
\Psi(P_{S}, P_{D})P_\beta \leq \lambda_\beta ,
\]
where $\beta \in \{S, D\}$, so that
\begin{equation}
\label{eq:app-ineq}
\Psi(P_{S}, P_{D}) \leq \min(\frac{\lambda_S}{P_S}, \frac{\lambda_D}{P_D}).
\end{equation}
Without loss of generality, assume $\frac{\lambda_S}{P_S} < \frac{\lambda_D}{P_D}$, then $\Psi(P_{S}, P_{D})P_D < \lambda_D$, which means $D$ always have enough energy to receive. The problem degenerates to calculating the transmission probability at $S$, and thus from Lemma \ref{lem:transprob}, $\Psi(P_{S}, P_{D}) = \frac{\lambda_S}{P_S}$.

For the special case of $\frac{\lambda_S}{P_S} = \frac{\lambda_D}{P_D}$, we must have $\Psi(P_{S}, P_{D}) = \frac{\lambda_S}{P_S} = \frac{\lambda_D}{P_D}$. Otherwise, according to (\ref{eq:app-ineq}), $\Psi(P_{S}, P_{D})P_\beta < \lambda_\beta$ ($\beta\in \{S, D\}$), which means that there will always be sufficient energy at both sides for transmission and receiving, which contradicts the fact that $\Psi(P_{S}, P_{D})<  \min(\frac{\lambda_S}{P_S}, \frac{\lambda_D}{P_D})<1$.

We thus complete the proof of the proposition.

\section{Proof of Theorem \ref{th:optimality2}}
\label{app:th:optimality2}
Denote $\Psi = \lim_{n \to +\infty} \frac{1}{n} \sum_{t = 1}^{n} \mathbf{E} \left[ \mathbf{1}_{B_{S}^{t} \geq P_{S}, B_{D}^{t} \geq P_{D} } \right]$ under any stationary power control policy not depending on $u$ at $S$. Due to the same reason as for (\ref{eq:app-ineq}), $\Psi \leq \frac{\lambda_D}{P_D}$. If $\lambda_D \geq P_D$, following the proof of Lemma \ref{lem:transprob}, $D$ always has enough energy, and thus according to Lemma \ref{lem:outageminEH}, the optimal policy at $S$ is threshold-based. And according to Theorem \ref{th:optimality1}, the optimal threshold $P_{S}^{*} = \max \left(\lambda_{S}, B_{\text{th}}\right)$.

If $\lambda_D < P_D$, consider a non-EH $S$ with average power constraint $\lambda_S$ first. Any policy at $S$ can be characterized as ``transmit with power consumption PDF $g(P)$ whenever $B_D^t \geq P_D$", as $S$ knows $B_D^t$. For the set of policies with the same $\Psi_D$ (the probability at any slot that $B_D^t \geq P_D$), denote the corresponding power PDF as $g(P)|_{\Psi_D}$. Following the similar proof procedure as in Lemma \ref{lem:outagemin}, the optimal $g(P)|_{\Psi_D}$ must have the same structure as (\ref{eq:optimalgp}). Therefore the optimal policy must be threshold-based. Denote the transmission threshold as $P_S$, following the same procedure in the proof of Proposition \ref{prop:jointprob}, we get the same transmission and receiving probability $\Psi_{P_S,P_D}$ as for the EH $S$ case. Therefore the optimal solutions of the two systems are the same, which indicates that for EH $S$, the threshold-based policy is optimal.

Finally, solving
\begin{equation}
\max_{P_{S} \geq 0} ~ \phi(P_{S}) \triangleq \exp \left( - \frac{(2^{R}-1)z}{P_{\text{tx}}} \right)  \min(1, \frac{\lambda_S}{P_S}, \frac{\lambda_D}{P_D})
\end{equation}
will get the optimal $P_S^*$ in the theorem.

\section{Proof of Proposition \ref{prop:recvprob2}}
\label{app:prop:recvprob2}
Similar to the proof of Lemma \ref{lem:transprob}, an energy packet queue $\mathcal{Q}$ is also constructed for node $D$. The difference is that, in $\mathcal{Q}$ whenever an energy packet, of volume $P_D$, is served by the ``server", some energy with amount of $(1-\xi)P_D$ gets back to the ``source" instantaneously with probability $1-\Psi_S$, where $\Psi_S = \min\{1,\frac{\lambda_S}{P_S}\}$ is the transmission probability of $S$. This is because $D$ only spends $\xi P_D$ when $S$ is not transmitting. Recall that only when $B_D^t\geq P_D$ will $D$ try to detect, and otherwise the detection is useless. As a result, the equivalent energy arrival rate is
\[
\lambda_D + (1-\xi)P_D(1-\Psi_S)\Psi_D.
\]
Following similar procedures of proving Lemma \ref{lem:transprob}, if $\lambda_D + (1-\xi)P_D(1-\Psi_S)\Psi_D \ge P_D$, then $\Psi_D = 1$. This requires $\lambda_D \ge [1-(1-\xi)(1-\Psi_S)]P_D$. Otherwise, if $\lambda_D + (1-\xi)P_D(1-\Psi_S)\Psi_D < P_D$, according to Little's Law \cite{Kleinrock}, we have
\begin{equation}
\Psi_D = \frac{\lambda_D + (1-\xi)P_D(1-\Psi_S)\Psi_D}{P_D},
\end{equation}
and so that
\begin{equation}
\Psi_D = \frac{\lambda_D}{[1-(1-\xi)(1-\Psi_S)]P_D}.
\end{equation}
This completes the proof.

\section{Proof of Proposition \ref{prop:recvprob3}}
\label{app:prop:recvprob3}
Similar to the proof of Proposition \ref{prop:recvprob2}, an energy packet queue $\mathcal{Q}$ is constructed for node $D$. But the mount of energy gets back to the ``source" is different. Recall that only when $B_D^t\geq P_F$ will $D$ try to detect and receive, and otherwise the detection is useless. As a result, the equivalent energy arrival rate is
\[
\lambda_D + \underbrace{(1-\xi\eta)P_F(1-\Psi_S)\Psi_D}_{\text{(a)}} + \underbrace{(1-\eta)P_Fp(P_{\text{tx}})\Psi_S\Psi_D}_{(b)},
\]
where $\Psi_S = \min\{1,\frac{\lambda_S}{P_S}\}$ is the transmission probability of $S$, and: (a) represents the case that $S$ is not transmitting, and thus $(1-\xi\eta)P_F$ of energy is put back to the energy queue; (b) represents the case that $S$ is transmitting but channel outage occurs (with probability $p(P_{\text{tx}})$), and thus $(1-\eta)P_F$ amount of energy is put back to the energy queue.
Following similar procedures of proving Lemma \ref{lem:transprob}, if $\lambda_D + (1-\xi\eta)P_F(1-\Psi_S)\Psi_D + (1-\eta)P_Fp(P_{\text{tx}})\Psi_S\Psi_D \ge P_F$, then $\Psi_D = 1$. This requires $\lambda_D \ge [1-(1-\xi\eta)(1-\Psi_S)-(1-\eta)p(P_{\text{tx}})\Psi_S]P_F$. Otherwise, according to Little's Law \cite{Kleinrock}, we have
\begin{equation}
\Psi_D \!=\! \frac{\lambda_D \!+\! (1\!-\!\xi\eta)P_F(1\!-\!\Psi_S)\Psi_D \!+\! (1\!-\!\eta)P_Fp(P_{\text{tx}})\Psi_S\Psi_D}{P_F},
\end{equation}
and so that
\begin{equation}
\Psi_D = \frac{\lambda_D}{[1-(1-\xi\eta)(1-\Psi_S)-(1-\eta)p(P_{\text{tx}})\Psi_S]P_F}.
\end{equation}
This completes the proof.

\section{Proof of Corollary \ref{cor:Psi2}}
\label{app:cor:Psi2}
When $D$ knows that $S$ is taking the disjoint threshold-based policy with threshold value $P_S$, $\Psi_D$ is different from (\ref{eq:recvprob}) in Lemma \ref{lem:recvprob}, while $\Psi_S$ remains the same as (\ref{eq:transprob}) in Lemma \ref{lem:transprob}. To calculate $\Psi_D$, similar to the proof of Proposition \ref{prop:recvprob2}, an energy packet queue $\mathcal{Q}$ is constructed for node $D$. But the mount of energy gets back to the ``source" is $P_D$ with probability $p(P_{\text{tx}})$, corresponding to the occurrence of channel outage. Because $D$ knows the threshold power $P_S$ and the CSI, if the channel is in outage, it is not necessary to spend any receiving energy. As a result, the equivalent energy arrival rate is
\[
\lambda_D + p(P_{\text{tx}})P_D\Psi_D.
\]
Following similar procedures of proving Lemma \ref{lem:transprob}, if $\lambda_D + p(P_{\text{tx}})P_D\Psi_D \ge P_D$, then $\Psi_D = 1$. This requires $\lambda_D \ge [1-p(P_{\text{tx}})]P_D$. Otherwise, if $\lambda_D + p(P_{\text{tx}})P_D\Psi_D < P_D$, according to Little's Law \cite{Kleinrock}, we have
\begin{equation}
\Psi_D = \frac{\lambda_D + p(P_{\text{tx}})P_D\Psi_D}{P_D},
\end{equation}
and so that
\begin{equation}
\Psi_D = \frac{\lambda_D}{[1-p(P_{\text{tx}})]P_D}.
\end{equation}
As a result, given $\Psi_S =  \min\left(1, \frac{\lambda_{S}}{P_S}\right)$, and since the events $B_S^t\ge P_S$ and $B_D^t \ge P_D$ are mutually independent, the corollary is proved.

\section{Proof of Corollary \ref{cor:optimality1}}
\label{app:cor:optimality1}
Denote the power control policy as a probability dense function (PDF) $g(P)$, and denote the per-slot transmission successful probability with transmission power $P$ as $f(P) = 1-p(P)$, where $p(P)$ is the channel outage probability defined in (\ref{eq:outage}). In this way, the average transmission successful probability is
\begin{equation}
p_{\text{suc}} =  \int_0^{+\infty} g(P) f(P) \text{d} P.
\end{equation}
Following the way of deriving $\Psi_D$ in Appendix \ref{app:cor:Psi2}, by replacing $1- p(P_{\text{tx}})$ (for disjoint threshold-based policy) with $p_{\text{suc}}$ (for power control policy described by $g(P)$), we get
\begin{equation}
\Psi_D =\min\left(1, \frac{\lambda_D}{p_{\text{suc}}P_D} \right).
\end{equation}
Then, the average successful probability $\phi$ becomes
\begin{equation}
\phi = p_{\text{suc}}\Psi_D = \min\left(p_{\text{suc}}, \frac{\lambda_D}{P_D} \right),
\end{equation}
where the first equality holds since the channel condition and the choice of the transmission power are independent, as $S$ only has CDI. From this equation, it is easy to see that maximizing $\phi$ is equivalent to maximizing $p_{\text{suc}}$. As a result, according to Lemma \ref{lem:outagemin} and \ref{lem:outageminEH}, the optimal policy at $S$ is threshold-based, and the optimal threshold $P_S^*$ is the same as the one in Theorem \ref{th:optimality1}.

\end{document}